\documentclass[journal]{IEEEtran}

\usepackage{amsthm}
\usepackage{amsmath}
\usepackage{amssymb}
\usepackage{amsmath}
\usepackage{amssymb}
\usepackage{epsfig}
\usepackage{epsf}
\usepackage{subfigure}
\usepackage{graphicx}
\usepackage{url}
\usepackage[]{authblk}

\def\BibTeX{{\rm B\kern-.05em{\sc i\kern-.025em b}\kern-.08em
    T\kern-.1667em\lower.7ex\hbox{E}\kern-.125emX}}

\newtheorem{claim}{Claim}

\newtheorem{theorem}{Theorem}
\newtheorem{definition}{Definition}

\newtheorem{lemma}{Lemma}

\newtheorem{remark}{Remark}

\author{Alireza Vahid, Vaneet Aggarwal, A. Salman Avestimehr, and Ashutosh Sabharwal\thanks{A. Vahid and A. S. Avestimehr are with the School of Electrical and Computer Engineering, Cornell University, Ithaca, NY (email: av292@cornell.edu, avestimehr@ece.cornell.edu). V. Aggarwal is with AT\&T Labs - Research, Florham Park, NJ 07932 (email: vaneet@research.att.com). A. Sabharwal is with Department of Electrical and Computer Engineering, Rice University, Houston, TX (email: ashu@rice.edu).

The results in this paper were presented in part at the Allerton Conference \cite{AlirezaAllerton}.
}
}

\begin{document}

\title{Wireless Network Coding with Local Network Views: Coded Layer Scheduling}

\maketitle

\begin{abstract}
One of the fundamental challenges in the design of distributed wireless networks is the large dynamic range of network state. Since continuous tracking of global network state at all nodes is practically impossible, nodes can only acquire limited local views of the whole network to design their transmission strategies. In this paper, we study multi-layer wireless networks and assume that each node has only a limited knowledge, namely \emph{1-local view}, where each S-D pair has enough information to perform optimally when other pairs do not interfere, along with connectivity information for rest of the network. We investigate the information-theoretic limits of communication with such limited knowledge at the nodes. We develop a novel transmission strategy, namely \emph{Coded Layer Scheduling}, that solely relies on 1-local view at the nodes and incorporates three different techniques:  (1) per layer interference avoidance, (2) repetition coding to allow overhearing of the interference, and (3) network coding to allow interference neutralization. We show that our proposed scheme can provide a significant throughput gain compared with the conventional interference avoidance strategies. Furthermore, we show that our strategy maximizes the achievable normalized sum-rate for some classes of networks, hence, characterizing the normalized sum-capacity of those networks with 1-local view.
\end{abstract}

\section{Introduction}
\label{introduction}

In dynamic wireless networks, optimizing system efficiency requires information about the state of the network in order to determine what resources are actually available.  However, in large wireless networks, keeping track of the state for making optimal decisions is typically infeasible. Thus, in the absence of centralization of network state information, nodes have limited local views of the network and make decentralized decisions based on their own local view of the network. The key question then is, how do optimal decentralized decisions perform in comparison to the optimal centralized decisions which rely on full network state information.

In this paper, we consider multi-source multi-destination multi-layer wireless networks and seek sum-rate optimal transmission strategies when sources have only limited local view of the network. To model local views at the nodes, we use a generalization of the hop-count based model that we introduced in ~\cite{VaneetIT} for single-layer networks. In the hop-count based model each source knows the channel gains of those links that are up to certain number of hops away and beyond that it only knows whether a link exists or not. The hop-count based model was appropriate for single-layer networks where all destinations are within one-hop from their respective sources. For multi-layer networks, a more scalable approach is to model local views based on the knowledge about source-destination (S-D) routes in the network (instead of source-destination links). The motivation for the route-based model stems from coordination protocols like routing which are often employed in multi-hop networks to discover S-D routes in the network. Hence, a reasonable quanta for network state information is the number of such end-to-end routes that are known at the source nodes. In this paper, we consider the case where each S-D pair has enough information to perform optimally when other pairs do not interfere. Beyond that, the only other information available at each node is the global network connectivity. We refer to this model of local network knowledge as $1$-local view.

Since each channel gain can range from zero to a maximum value, our formulation is similar to compound channels \cite{blackwell,viswanath} with one major difference. In the multi-terminal compound network formulations, {\it all} nodes are missing identical information about the channels in the network, whereas, in our formulation, the $1$-local view results in {\it asymmetric} information about channels at different nodes.

In this paper, our metric to measure the performance of transmission strategies is {\it normalized sum-capacity} as defined in \cite{VaneetIT}, which represents the maximum fraction of the sum-capacity with full knowledge that can be always achieved when nodes only have partial knowledge about the network.

\subsection{Contributions}

Our main contribution is a new transmission scheme, named Coded Layer (CL) Scheduling, which only requires  $1$-local view at the nodes and combines coding with interference avoidance scheduling. Developed as a graph coloring algorithm on a route-extended graph, coded layer scheduling is a combination of three main techniques: (1) per layer interference avoidance, (2) repetition coding to allow overhearing of the interference, and (3) network coding to allow interference neutralization. 

We characterize the achievable normalized sum-rate of the CL scheduling as the solution to a new graph coloring problem and analyze its optimality for some classes of networks. In particular, we show that coded layer scheduling achieves the normalized sum-capacity in single-layer and two-layer $(K,m)$-folded chain networks (defined in Section~\ref{performance}). Furthermore, by considering $L$-nested folded-chain networks (defined in Section~\ref{performance}), we show that the gain from CL scheduling over  interference avoidance scheduling can be unbounded. 

We also investigate network topologies in which with $1$-local view at the nodes, interference avoidance scheduling is information-theoretically optimal. More specifically, we consider another class of networks, \emph{i.e.} $K \times \underbrace{2 \times \ldots \times 2}_M \times K$ networks, which is a $K$-flow network where all intermediate layers have only $2$ relays. We show that for this class, a simpler scheme based on only interference avoidance techniques, named Independent Layer (IL) scheduling, is optimal and coding is not required to achieve normalized sum-capacity with $1$-local view. In our limited experience, coding across time-slots can provide gains when there is some regular topological structure in the network and/or there is dense connecvitity.  However, the general connection between network topology, partial information and optimal schemes remains a largely open problem.

\subsection{Related Work}

In any network state learning algorithm, network state information is obtained via a form of message passing between the nodes. Since the channels through which the communication takes place are noisy and have delay, imprecise network information at the nodes becomes an important problem. Many models for imprecise network information have been considered for interference networks. These models range from having no channel state information at the sources \cite{V2009,Guo2010,jaf-blind,jafar-nocsit,vish-nocsit}, delayed channel state information \cite{retro,khand-delay,varanasi-delay,tse-delay} or analog feedback of channel state for fully-connected interference channels \cite{heath}. Most of these works assume fully connected network or a small number of users. A study to understand the role of limited network knowledge, was first initiated in \cite{messageisit,messageit} for general single-layer networks with {\it arbitrary} connectivity, where the authors used a message-passing abstraction of network protocols to formalize the notion of local view of the network at each node, such that the view at different nodes are mismatched from each others'. The key result was that local-view-based (decentralized) decisions can be either sum-rate optimal or can be arbitrarily worse than the global-view (centralized) sum-capacity.

The initial work in~\cite{messageisit,messageit} was strengthened for arbitrary $K$-user single-layer interference network in \cite{VaneetIT,aas09allerton,VaneetISIT}, where the authors proposed a new metric, normalized sum-capacity, to measure the performance of distributed decisions. Further,  the authors computed the normalized sum-capacity of distributed decisions for several network topologies with one-hop, two-hop and three-hop local-view information at each source. In this paper, we investigate the performance of decentralized decisions for multi-layer wireless networks.

The rest of the paper is organized as follows. In section \ref{problem}, we will introduce our network model and the new model to capture partial network knowledge and we define the notion of normalized sum-capacity. In Section~\ref{section:examples}, via a number of examples, we motivate our transmission strategies. In Section~\ref{section:CL}, we present our main result, \emph{i.e.} coded layer scheduling, and we charaterize its performance for multi-layer networks. In Section~\ref{performance}, we prove the optimality of our strategies (in terms of achieving normalized sum-capacity) for some networks. Finally, Section~\ref{conclusion} concludes the paper and presents some future directions. 

\section{Problem Formulation}
\label{problem}

In this section, we introduce our models for channel, network, and network knowledge at the nodes. We further define the notions of normalized sum-capacity introduced in \cite{VaneetIT}, which will be used to measure the performance of the strategies with partial network knowledge. 

\subsection{Network Model and Notations}
\label{subsec:netModel}

In this subsection, we will describe two channel models that will be studied in the paper, namely the linear deterministic model~\cite{ADTJ10}, and the Gaussian model. In both models, a network is represented by a directed graph $\mathcal{G} = (\mathcal{V}, \mathcal{E},\{ w_{ij} \}_{(i,j) \in \mathcal{E}})$, where $\mathcal{V}$ is the set of vertices representing nodes in the network, $\mathcal{E}$ is the set of directed edges representing links among the nodes, and  $\{ w_{ij} \}_{(i,j) \in \mathcal{E}}$ represents the channel gains associated with the edges.

We consider a layered network in this paper, \emph{i.e.} the nodes in this network can be partitioned into $L$ subsets $\mathcal{V}_1, \mathcal{V}_2, \ldots, \mathcal{V}_L$. Out of $|\mathcal{V}|$ nodes in the network, $K$ are denoted as sources and $K$ are destinations. We label these source and destination nodes by ${\sf S}_i$ and ${\sf D}_i$ respectively,  $i = 1,2,\ldots,K$. We set $\mathcal{V}_1 = \{ {\sf S}_1, {\sf S}_2, \ldots, {\sf S}_K \}$ and $\mathcal{V}_L = \{ {\sf D}_1, {\sf D}_2, \ldots, {\sf D}_K \}$. The remaining $|\mathcal{V}|-2K$ nodes are relay nodes which facilitate the communication between sources and destinations. We denote a specific relay in $\mathcal{V}_l$ by ${\sf V}^l_i$, $i = 1,2, \ldots, |\mathcal{V}_l|$ and $l = 2,3,\ldots, L-1$. Without loss of generality, we can also refer to a node in $\mathcal{V}$ simply as ${\sf V}_i$, $i = 1,2,\ldots,|\mathcal{V}|$.

The layered structure of the network imposes the following constraint on the edges in the network, 
\begin{align}
(i,j) \in \mathcal{E} &\Rightarrow \exists l \in \{1,2,\ldots,L-1\} \text{ such that: } \nonumber \\
\; &\left( {\sf V}_i \in \mathcal{V}_l \text{ and } {\sf V}_j \in \mathcal{V}_{l+1} \right).
\end{align}

The two channel models used in this paper are as follows.
\begin{enumerate}

\item
\emph{The Linear Deterministic Model}~\cite{ADTJ10}: In this model, there is a non-negative integer, $w_{ij} = n_{ij}$, associated with each link $(i,j) \in \mathcal{E}$, which represents its gain. Let $q$ be the maximum of all the channel gains in this network. In the linear deterministic model, the channel input at node ${\sf V}_i$ at time $t$ is denoted by $X_{{\sf V}_i}[t] = [ X_{{\sf V}_{i_1}}[t], X_{{\sf V}_{i_2}}[t], \ldots, X_{{\sf V}_{i_q}}[t] ]^T \in \mathbb{F}_2^q$. The received signal at node ${\sf V}_j$ at time $t$ is denoted by $Y_{{\sf V}_{j}}[t] = [ Y_{{\sf V}_{j_1}}[t], Y_{{\sf V}_{j_2}}[t], \ldots, Y_{{\sf V}_{j_q}}[t] ]^T \in \mathbb{F}_2^q$, and is given by
\begin{equation}
Y_{{\sf V}_{j}}[t] = \sum_{i: (i,j) \in \mathcal{E}}{{\bf S}^{q-n_{ij}} X_{{\sf V}_{i}}[t]},
\end{equation}
where ${\bf S}$ is the $q \times q$ shift matrix and the operations are in $\mathbb{F}_2^q$. If a link between ${\sf V}_i$ and ${\sf V}_j$ does not exist, we set $n_{ij}$ to be zero.

\item
\emph{The Gaussian Model:} In this model, the channel gain $w_{ij}$ is denoted by $h_{ij} \in \mathbb{C}$. The channel input at node ${\sf V}_i$ at time $t$ is denoted by $X_{{\sf V}_{i}}[t] \in \mathbb{C}$, and the received signal at node ${\sf V}_j$ at time $t$ is denoted by $Y_{{\sf V}_{j}}[t] \in \mathbb{C}$ given by
\begin{equation}
Y_{{\sf V}_{j}}[t] = \sum_i{h_{ij} X_{{\sf V}_{i}}[t] + Z_j[t]},
\end{equation}
where $Z_j[t]$ is the additive white complex Gaussian noise with unit variance. We also assume a power constraint of $1$ at all nodes, \emph{i.e.} $\lim_{n\to\infty}\frac{1}{n} \mathbb{E}(\sum_{t=1}^n{|X_{{\sf V}_{i}}[t]|^2}) \le 1$.

\end{enumerate}

A \emph{route} from a source ${\sf S}_i$ to a destination ${\sf D}_j$ is a set of nodes such that there exists an ordering of these nodes where the first one is ${\sf S}_i$, last one is ${\sf D}_j$, and any two consecutive nodes in this ordering are connected by an edge in the graph. 

\begin{definition}
An {\it induced subgraph $\mathcal{G}_{ij}$} is a subgraph of $\mathcal{G}$ with its vertex set being the union of all routes from source ${\sf S}_i$ to a destination ${\sf D}_j$, and its edge set being the subset of all edges in $\mathcal{G}$ between the vertices of $\mathcal{G}_{ij}$.
\end{definition}

We say that S-D pair $i$ and S-D $j$ are \emph{non-interfering} if $\mathcal{G}_{ii}$ and $\mathcal{G}_{jj}$ are two disjoint induced subgraphs of $\mathcal{G}$.

The in-degree function $d_{\mathrm{in}}({\sf V}_i)$, is the number of in-coming edges connected to node ${\sf V}_i$. Similarly, the out-degree function $d_{\mathrm{out}}({\sf V}_i)$, is the number of out-going edges connected to node ${\sf V}_i$. Note that the in-degree of a source and the out-degree of a destination are both equal to $0$. The maximum degree of the nodes in $\mathcal{G}$ is defined as
\begin{equation}
d_{\max} = \max_{i \in \{ 1,\ldots,|\mathcal{V}| \}} \left( d_{\mathrm{in}}({\sf V}_i), d_{\mathrm{out}}({\sf V}_i) \right).
\end{equation}

We also need the following definitions that will be used later in this paper.
\begin{definition}
At any node ${\sf V}_i \in \mathcal{G}$, we define the index set $\mathcal{J}_{{\sf V}_i}$ as follows
\begin{equation}
\mathcal{J}_{{\sf V}_i} := \{ j | {\sf V}_i \in \mathcal{G}_{jj}, j = 1,2,\ldots,K \}.
\end{equation}
In other words, $\mathcal{J}_{{\sf V}_i}$ is the set of indices of those S-D pairs that have ${\sf V}_i$ on a route between them.
\end{definition}

\begin{definition}
The \emph{route-expanded graph} $\mathcal{G}_{\sf exp} = \left( \mathcal{V}_{\sf exp}, \mathcal{E}_{\sf exp} \right)$ associated with a layered network $\mathcal{G} = (\mathcal{V}, \mathcal{E},\{ w_{ij} \}_{(i,j) \in \mathcal{E}})$ with sources in $\mathcal{V}_1 = \{ {\sf S}_1, {\sf S}_2, \ldots, {\sf S}_K \}$ and destinations in $\mathcal{V}_L = \{ {\sf D}_1, {\sf D}_2, \ldots, {\sf D}_K \}$ is constructed by replacing each node ${\sf V}_i \in \mathcal{V}$ with $|\mathcal{J}_{{\sf V}_i}|$ nodes represented by ${\sf V}_{i,j}$ where $j \in \mathcal{J}_{{\sf V}_i}$, and connect them according to
\begin{align*}
 ({\sf V}_{i,j},{\sf V}_{i^\prime,j^\prime}) \in \mathcal{E}_{\sf exp} \text{ iff } (i,i^\prime) \in \mathcal{E}.
\end{align*}
\end{definition}

We define $\bar{\mathcal{V}}_i = \{ {\sf V}_{i,j} | j \in \mathcal{J}_{{\sf V}_i} \}$, $i = 1,2,\ldots,|\mathcal{V}|$ (\emph{i.e.} all the duplicates of node ${\sf V}_i$) and we refer to it as a super-node (or equivalently a super-relay if ${\sf V}_i$ is a relay). 

For an illustration of the route-expanded graph see Figure~\ref{fig:route-expanded}. For simplicity, we have represented each pair with a shape, \emph{i.e.}  \raisebox{3pt}{\circle{6}}, $\triangle$, and $\Box$ for S-D pairs $1$, $2$, and $3$ respectively. The route-expanded graph of the network in Figure~\ref{fig:route-expanded}(a) is illustrated in Figure~\ref{fig:route-expanded}(b). Each relay is on a route for two S-D pairs, hence each super-relay contains two nodes.

\begin{figure}[ht]
\centering
\subfigure[]{\includegraphics[height = 3cm]{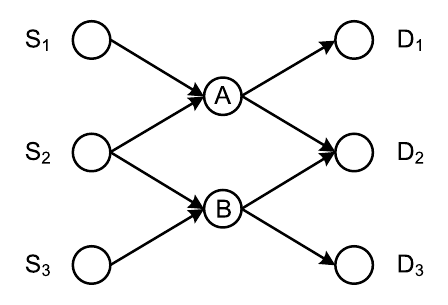}}
\hspace{0.3in}
\subfigure[]{\includegraphics[height = 3cm]{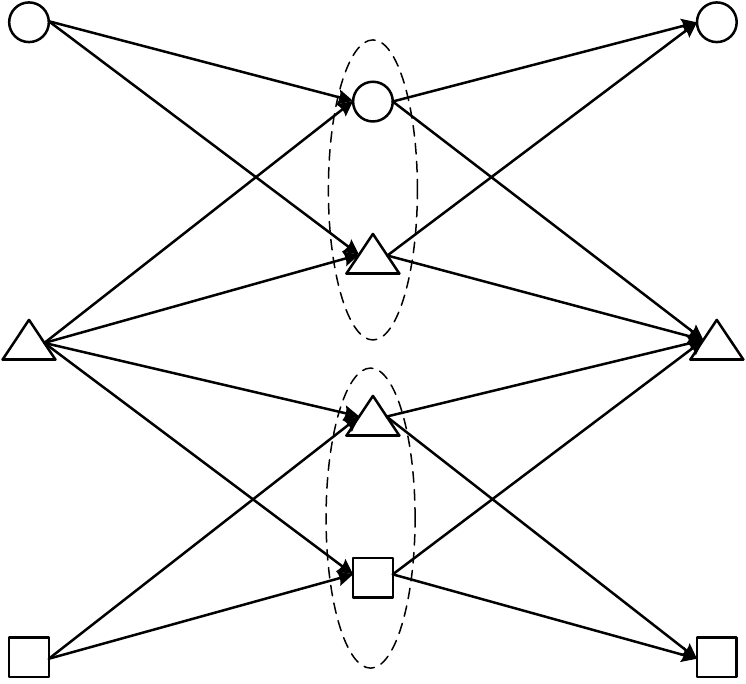}}
\caption{\it (a) A $2$-layer network, and (b) its route-expanded graph.\label{fig:route-expanded}}
\end{figure}

\subsection{Model of Partial Network Knowledge}
\label{localview}

In this subsection, we describe the model of \cite{AlirezaAllerton} for partial network information that will be used in this paper. We first define the \emph{route-adjacency graph} $\mathcal{G}^\prime$ of $\mathcal{G}$, which is an undirected bipartite graph consisting of all sources on one side and all destinations on the other side; see Figure~\ref{adjacency} for an example. A source ${\sf S}_i$ and a destination ${\sf D}_j$ are connected in $\mathcal{G}^\prime$, if there exists a route between them in $\mathcal{G}$. More formally, $\mathcal{G}^\prime = (\mathcal{V}^\prime,\mathcal{E}^\prime)$ where $\mathcal{V}^\prime = \mathcal{V}_1 \cup \mathcal{V}_L$ and $\mathcal{E}^\prime = \{ (i,j) | \exists \text{ a route from } {\sf S}_i \text{ to } {\sf D}_j \}$.

We now define the model for partial network knowledge that will be used in the paper,  namely \emph{$h$-local view}, as the following:

\begin{itemize}
\item All nodes have full knowledge of the network topology, $(\mathcal{V},\mathcal{E})$, i.e., which links are in $\mathcal{G}$, but not their channel gains. The network topology knowledge is denoted by side information $\mathsf{SI}$.
\item Each source, ${\sf S}_i$, knows the gains of all those channels that are in a route from source ${\sf S}_j$ to destination ${\sf D}_k$, such that  ${\sf S}_j$ and  ${\sf D}_k$  are at most $h$ hops away from ${\sf S}_i$ in $\mathcal{G}^\prime$. The $h$-hop channel knowledge at a source is denoted by $L_{{\sf S}_i}$.
\item  Each node ${\sf V}_i$ (which is not a source) has the union of the information of all those sources that have a route to it, and this knowledge at node is denoted by $L_{{\sf V}_i}$. 
\end{itemize}

Note that this model is a generalization of the hop-based model for partial network knowledge in single layer networks~\cite{VaneetIT}. While the partial information model is general, we will focus on the case where $h=1$. In other words, each S-D pair has enough information to perform optimally when other pairs do not interfere (i.e., it knows the channel gains of all links that are in a route to its own destination). However beyond that, each pair only knows the connectivity in the network (structure of interference). We are interested to find if one can outperform interference avoidance techniques with such limited knowledge. In the following subsection, we define the metrics we use to measure the performance of transmission strategies with $h$-local view.

\begin{figure}[hb]
\centering
\subfigure[]{\includegraphics[height = 3cm]{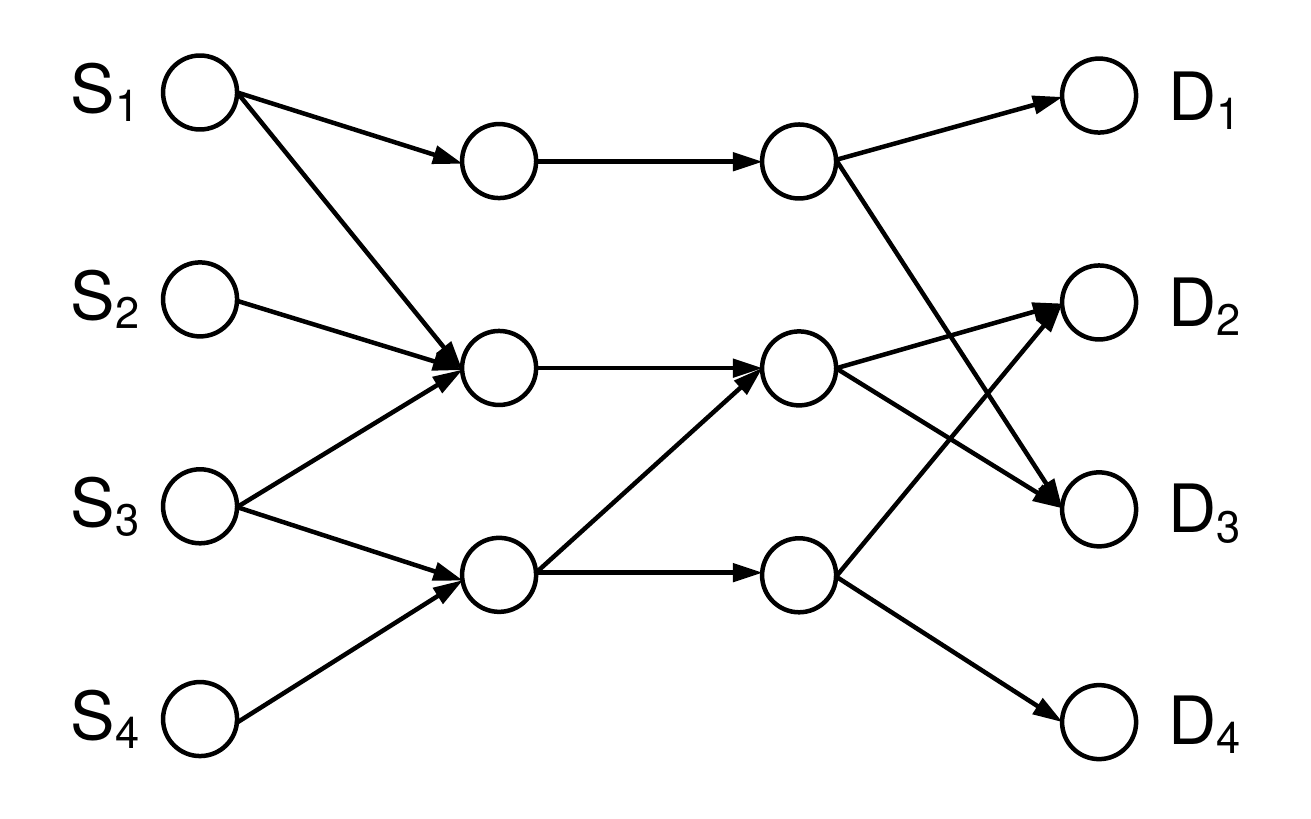}}
\hspace{0.3in}
\subfigure[]{\includegraphics[height = 3cm]{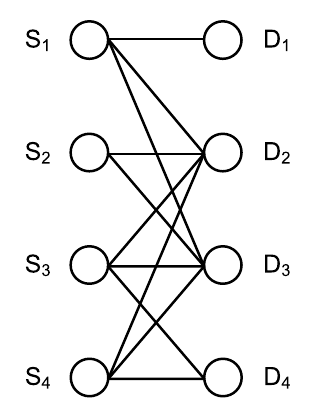}}
\caption{\it (a) A multi-layer network, and (b) its route-adjacency graph. \label{adjacency}}
\end{figure}

\subsection{Normalized Sum-Capacity}

We now define the notion of normalized sum-capacity, which is our metric for evaluating network capacity with partial network knowledge \cite{VaneetIT,VaneetISIT}. Normalized sum-capacity represents the maximum fraction of the sum-capacity with full knowledge that can be always achieved when nodes only have partial knowledge about the network, and is defined as follows.


Consider the scenario in which source ${\sf S}_i$ wishes to reliably communicate message $\hbox{W}_i \in \{ 1,2,\ldots,2^{N R_i}\}$ to destination ${\sf D}_i$ during $N$ uses of the channel, $i=1,2,\ldots,K$. We assume that the messages are independent and chosen uniformly. For each source ${\sf S}_i$, let message $\hbox{W}_i$ be encoded as $X_{{\sf S}_i}^N$ using the encoding function $e_{i}(\hbox{W}_i|L_{{\sf S}_i},\mathsf{SI})$, which depends on the available local network knowledge, $L_{{\sf S}_i}$, and the global side information, $\mathsf{SI}$.

Each relay in the network creates its input to the channel $X_{{\sf V}_i}$, using the encoding function $f_{{\sf V}_i}[t](Y_{{\sf V}_i}^{(t-1)}|L_{{\sf V}_i},\mathsf{SI})$, which depends on the available network knowledge, $L_{{\sf V}_i}$, and the side information, $\mathsf{SI}$, and all the previous received signals at the relay $Y_{{\sf V}_i}^{(t-1)} = \left[ Y_{{\sf V}_i}[1], Y_{{\sf V}_i}[2], \ldots, Y_{{\sf V}_i}[t-1] \right]$. A relay strategy is defined as the union of of all encoding functions used by the relays, $\{ f_{{\sf V}_i}[t](Y_{{\sf V}_i}^{(t-1)}|L_{{\sf V}_i},\mathsf{SI}) \}$, $t=1,2,\ldots,N$ and ${\sf V}_i \in \bigcup_{j=1}^{L-1}{\mathcal{V}_j}$.

Destination ${\sf D}_i$ is only interested in decoding $\hbox{W}_i$ and it will decode the message using the decoding function $\widehat{\hbox{W}}_i = d_{i}(Y_{{\sf D}_i}^N |L_{{\sf D}_i},\mathsf{SI})$, where $L_{{\sf D}_i}$ is the destination ${\sf D}_i$'s network knowledge. Note that the local view can be different from node to node.

\begin{definition}
A \emph{Strategy} $\mathcal{S}_N$ is defined as the set of: (1) all encoding functions at the sources; (2) all decoding functions at the destinations; and (3) the relay strategy for $t=1,2,\ldots,N$, \emph{i.e.}
\begin{equation}
\label{strategydef}
\mathcal{S}_N =
\left\{ \begin{array}{ll}
e_{i}(\hbox{W}_i|L_{{\sf S}_i},\mathsf{SI}) & i = 1,2,\ldots,K \\
f_{{\sf V}_i}[t](Y_{{\sf V}_i}^{(t-1)}|L_{{\sf V}_i},\mathsf{SI}) & t=1,2,\ldots,N \\
& \text{and } {\sf V}_i \in \bigcup_{j=1}^{L-1}{\mathcal{V}_j} \\
d_{i}(Y_{{\sf D}_i}^N |L_{{\sf D}_i},\mathsf{SI}) & i = 1,2,\ldots,K
\end{array} \right\}.
\end{equation}
\end{definition}

An error occurs when $\widehat{\hbox{W}}_i \neq \hbox{W}_i$ and we define the decoding error probability, $\lambda_i$, to be equal to $P(\widehat{\hbox{W}}_i \neq \hbox{W}_i)$. A rate tuple $(R_1,R_2, \ldots,R_K)$ is said to be achievable, if there exists a set of strategies $\{ \mathcal{S}_j \}_{j=1}^N$ such that the decoding error probabilities $\lambda_1,\lambda_2,\ldots,\lambda_K$ go to zero as $N \rightarrow \infty$ for all network states consistent with the side information. Moreover, for any S-D pair $i$, denote the maximum achievable rate $R_i$ with full network knowledge by $C_i$. The sum-capacity $C_{\mathrm{sum}}$, is the supremum of $\sum_{i=1}^K{R_i}$ over all possible encoding and decoding functions with full network knowledge.


We will now define the normalized sum-rate and the normalized sum-capacity.

\begin{definition}[\cite{VaneetIT}]
\emph{Normalized sum-rate} of $\alpha$ is said to be achievable, if there exists a set of strategies $\{ \mathcal{S}_j \}_{j=1}^N$ such that following holds. As $N$ goes to infinity, strategy $\mathcal{S}_N$ yields a sequence of codes having rates $R_i$ at the source ${\sf S}_i$, $i=1,\ldots,K$, such that the error probabilities at the destinations, $\lambda_1, \cdots \lambda_K$, go to zero, satisfying
\[
\sum_{i=1}^K{R_i} \ge \alpha C_{\mathrm{sum}} -\tau
\]
for all the network states consistent with the side information, and for a constant $\tau$ that is independent of the channel gains.
\end{definition}

\begin{definition}[\cite{VaneetIT}]
\emph{Normalized sum-capacity} $\alpha^*$, is defined as the supremum of all achievable normalized sum-rates $\alpha$. Note that $\alpha^* \in [0,1]$.
\end{definition} 

\section{Motivating Examples}
\label{section:examples}

Before diving into the main results in Section~\ref{section:CL}, we will use a sequence of examples to arrive at the main ingredients of the proposed coded layer scheduling. The key point of the discussion is to understand the mechanisms that allow outperforming interference avoidance with only $1$-local view.


As defined earlier, $1$-local view means that each S-D pair has enough information to perform optimally when other pairs do not interfere. However, beyond 1-local view, each node only knows the connectivity in the network (structure of interference). So, at first glance it seems that the optimal strategy is to avoid interference between the S-D pairs and at each time, schedule as many non-interfering pairs as possible. Through an example, we investigate the performance of the above strategy which maximizes spatial reuse while avoiding interference at each node.

\begin{figure}[ht]
\centering
\subfigure[]{\includegraphics[height = 2.5cm]{MIP.pdf}}
\hspace{0.4in}
\subfigure[]{\includegraphics[height = 2.5cm]{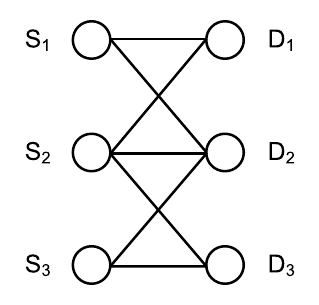}}
\caption{\it (a) A network where interference avoidance between S-D pairs is optimal, and (b) its route-adjacency graph.\label{MIP}}
\end{figure}

Consider the network depicted in Figure \ref{MIP}(a) with $1$-local view. From the route-adjacency graph of this network depicted in Figure~\ref{MIP}(b), we can see that S-D pairs $1$ and $3$ are non-interfering. We implement an achievability strategy described as follows. We split the communication block into two time-slots of equal length and represent each time-slot with a color, namely black and white. S-D pairs $1$ and $3$ communicate over time-slot black, whereas, S-D pair $2$ communicate over time-slot white. With this coloring, we have effectively seperated induced subgraphs of interfering pairs, see Figure~\ref{fig:MIPcoloring}. Now, since each pair can communicate interference-free over half of the communication block length, it can achieve half of its capacity with full network knowledge. Hence, we achieve a normalized sum-rate of $\alpha = 1/2$.

\begin{figure}[ht]
\centering
\subfigure[]{\includegraphics[height = 2.5cm]{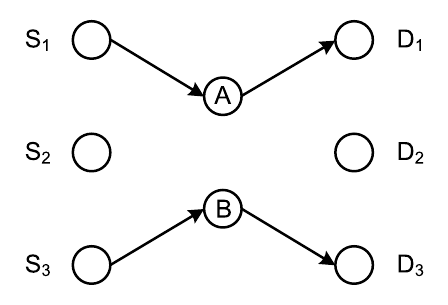}}
\hspace{0.4in}
\subfigure[]{\includegraphics[height = 2.5cm]{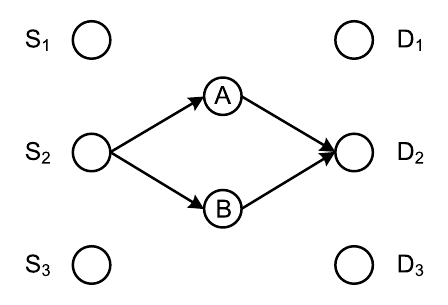}}
\caption{\it (a) S-D pairs $1$ and $3$ can simultaneously communicate interference-free over their induced subgraphs in the first time-slot (black time-slot), and (b) S-D pair $2$ can communicate interference-free over its induced subgraph in the second time-slot (white time-slot).\label{fig:MIPcoloring}}
\end{figure}

This scheduling strategy can be viewed as a specific coloring of nodes in the route-expanded graph (defined in Section~\ref{problem}). Consider the route-expanded graph of this example, as shown in Figure~\ref{fig:MIPexpanded}. The aforementioned scheduling strategy can be viewed as a coloring of nodes in the route-extended graph, such that (1) all nodes of the same shape, \emph{i.e.} pair ID, receive the same color, and (2) any two nodes with different shapes, \emph{i.e.} different pair IDs, that are connected to each other should have different colors. Note that since the nodes inside the same super-node are connected to the same nodes in $\mathcal{V}_{\sf exp}$ and they have different shapes, they will be assigned different colors. Figure~\ref{fig:MIPexpanded} illustrates such coloring of nodes in this example by using only two colors, B and W. In other words, we have assigned different colors to the induced subgraphs of interfering S-D pairs. Therefore, each S-D pair gets a chance to communicate over its induced subgraph interference-free during the time-slot associated with its color.

\begin{figure}[ht]
\centering
\includegraphics[height = 4cm]{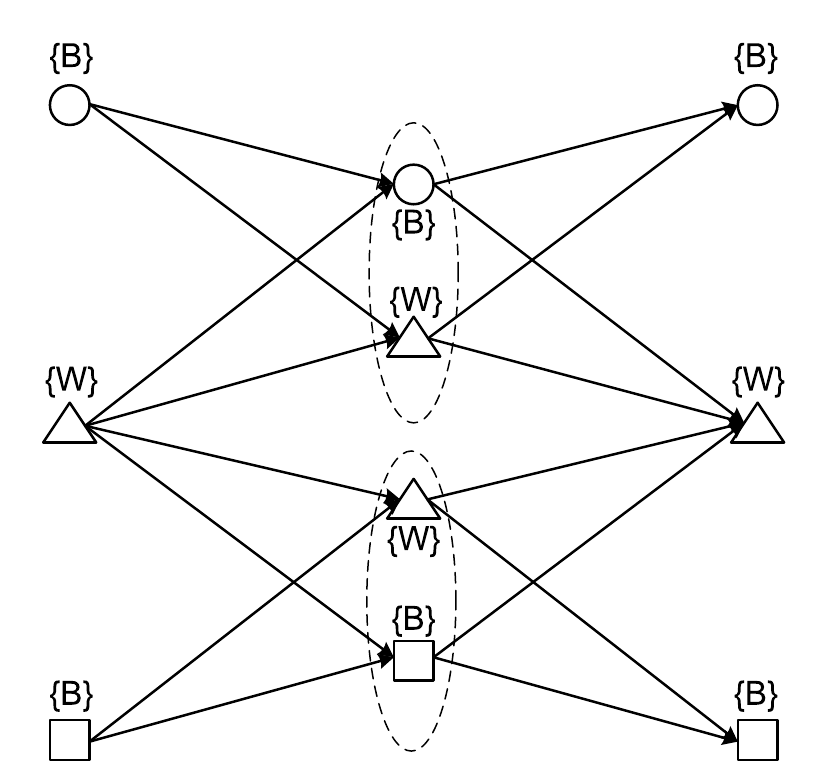}
\caption{\it Route-expanded graph, $\mathcal{G}_{\sf exp}$, of the example in Figure~\ref{MIP}(a).
\label{fig:MIPexpanded}}
\end{figure}

As we will see in the next lemma, which is proved in Appendix~\ref{Appendix:generalhalf}, $\alpha = 1/2$ is also an upper bound on the normalized sum-capacity of this network. Hence, scheduling non-interfering pairs performs optimally in this example. More generally, the following upper bound on $\alpha$ exists for a general class of multi-layer networks.

\begin{lemma}
\label{lemma:generalhalf}
In a $K$-user multi-layer network (linear deterministic or Gaussian) with $1$-local view, if there exists a path from ${\sf S}_i$ to ${\sf D}_j$, for some $i \neq j$, then the normalized sum-capacity is upper-bounded by $\alpha = 1/2$.
\end{lemma}

While the aforementioned interference avoidance strategy performed optimally in the network depicted in Figure~\ref{MIP}(a), we now illustrate an example where it is \emph{not} optimal. A key observation is that the scheduling described above, ignores the available knowledge of interference structure in each layer of the network and it only schedules pairs that are non-interfering over all layers. To see how this knowledge of interference structure can be exploited, consider the network depicted in Figure~\ref{MIL}(a) with $1$-local view. Applying the previous scheduling to this network, we achieve a normalized sum-rate of $\alpha = \frac{1}{3}$. However, we show that it is possible to go beyond $\alpha=\frac{1}{3}$ and achieve a normalized sum-rate of $\alpha=\frac{1}{2}$ for Figure~\ref{MIL} example.

\begin{figure}[t]
\centering
\subfigure[]{\includegraphics[height = 2.5cm]{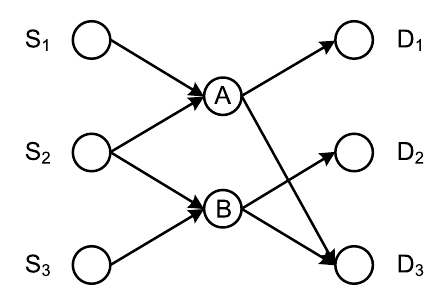}}
\hspace{0.4in}
\subfigure[]{\includegraphics[height = 2.5cm]{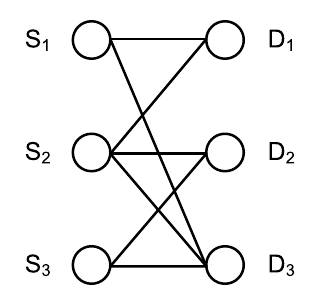}}
\caption{\it (a) A network where end-to-end interference avoidance is not optimal, and (b) its route-adjacency graph. \label{MIL}}
\end{figure}

We implement an achievability strategy described as follows. Similar to the previous example, we split the communication block into two time-slots of equal length and represent each time-slot with a color, namely black and white. Unlike the previous case where we assigned a color to each S-D pair, in this example, the color assignment is carried on in each layer seperately. We let sources $1$ and $3$ to communicate in the first layer over time-slot black and source $2$ over time-slot white. However, in the second layer, relays communicate to destinations ${\sf D}_1$ and ${\sf D}_2$ in time-slot black and to destination ${\sf D}_3$ in time-slot white, see Figures \ref{fig:MILcoloring}(a) and \ref{fig:MILcoloring}(b). With this strategy each S-D pair can communicate over its induced subgraph interference-free during half of the communication block length, see Figures \ref{fig:MILcoloring}(c), \ref{fig:MILcoloring}(d) and \ref{fig:MILcoloring}(e). Hence, we achieve a normalized sum-rate of $\alpha = 1/2$. By Lemma~\ref{lemma:generalhalf}, we know that $\alpha=\frac{1}{2}$ is also an upper-bound on the normalized sum-rate of this network, hence, we have achieved it normalized sum-capacity.

\begin{figure}[t]
\centering
\subfigure[]{\includegraphics[height = 2.5cm]{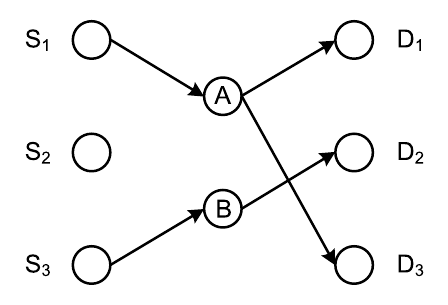}}
\hspace{0.4in}
\subfigure[]{\includegraphics[height = 2.5cm]{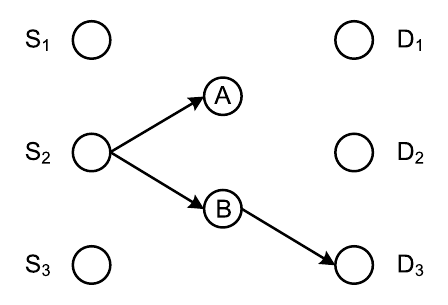}}

\subfigure[]{\includegraphics[height = 2.5cm]{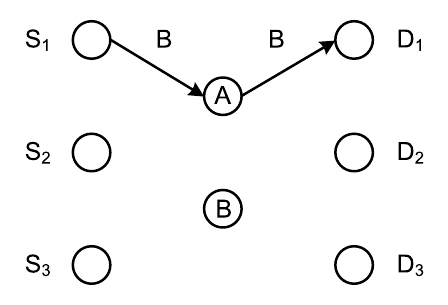}}
\hspace{0.4in}
\subfigure[]{\includegraphics[height = 2.5cm]{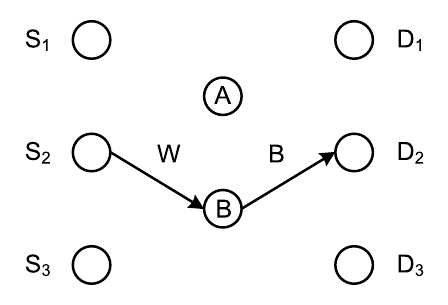}}

\subfigure[]{\includegraphics[height = 2.5cm]{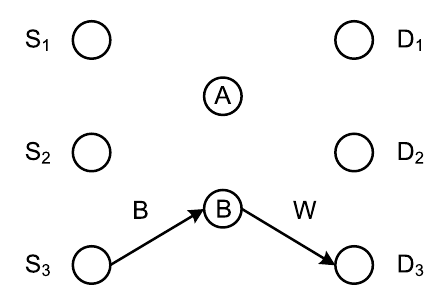}}

\caption{\it (a) Sources $1$ and $3$ can simultaneously communicate interference-free over time-slot black, and relays can communicate with destinations $1$ and $2$ interference-free over time-slot black, (b) Source $2$ can communicate interference-free over time-slot white, and destination $3$ can receive its signal interference-free over the same time-slot, and (b), (c), and (d) the interfernce-free induced subgraphs of S-D pair $1$, $2$, and $3$ respectively. \label{fig:MILcoloring}}
\end{figure}

This strategy can be viewed as a modification of our previous coloring of the nodes in the route-expanded graph as follows. Nodes with the same shape can be assigned different colors at different layers, however, still any two nodes with different shapes that are connected to each other should have different colors. In other words, we assign colors such that the induced subgraphs of different S-D pairs have different colors in each layer only if they are interfering at that layer. Figure~\ref{fig:MILexpanded} illustrates such coloring of nodes in this example by using only two colors, B and W. Since the induced subgraphs of interfering pairs have different colors in each layer, each S-D pair has a chance to communicate over its induced subgraph interference-free during half of the communication block.

\begin{figure}[ht]
\centering
\includegraphics[height = 4cm]{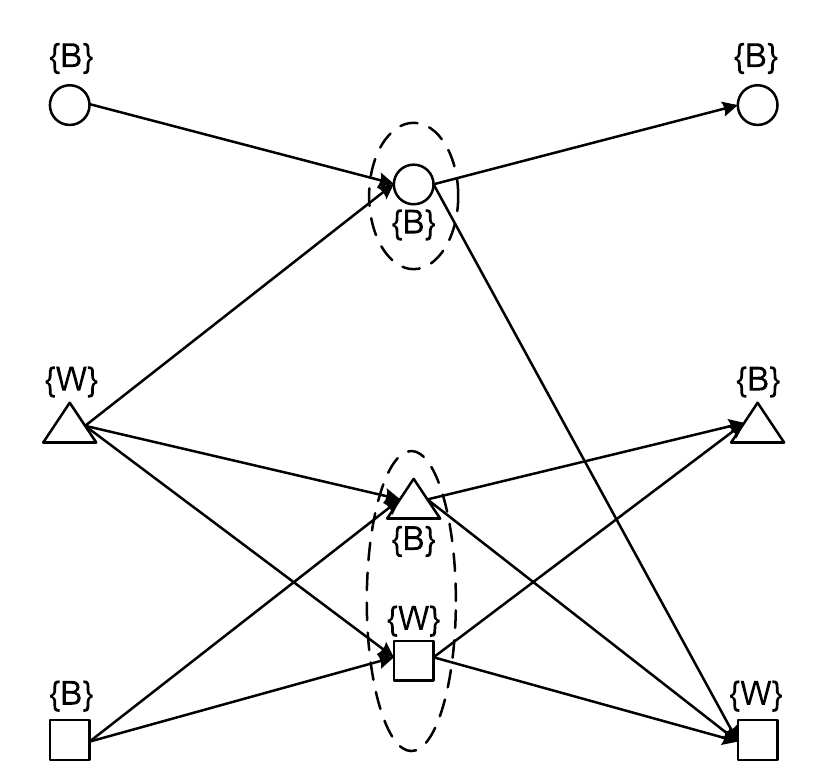}
\caption{\it (a) Route-expanded graph, $\mathcal{G}_{\sf exp}$, of the network depicted in Figure~\ref{MIL}(a).\label{fig:MILexpanded}}
\end{figure}

The scheduling developed for the network depicted in Figure~\ref{MIL}, illustrates a major deficiency of the scheduling developed for the network depicted in Figure~\ref{MIP}, which is the restriction of applying the same scheduling to all nodes on a route between ${\sf S}_i$ and ${\sf D}_i$. By exploiting the available information of interference structure and scheduling nodes in different layers separately, we outperformed the first scheduling. In Section~\ref{section:CL}, we will formally define this new scheme and refer to it as Maximal Independent Layer (MIL) scheduling. 

So far, our proposed transmission strategies are based on interference avoidance either in an end-to-end manner or in a per-layer manner (MIL scheduling). But, can we go beyond interference avoidance with such limited knowledge at the nodes? To answer this question, first consider the single-layer network depicted in Figure~\ref{fig:33}(a). Since the conflict graph of this network is fully connected, using MIL scheduling we can only achieve $\alpha=\frac{1}{3}$. However, we now show that it is possible to achieve $\alpha=\frac{1}{2}$ by employing a coding strategy that only requires $1$-local view.

\begin{figure}[ht]
\centering
\subfigure[]{\includegraphics[height = 3cm]{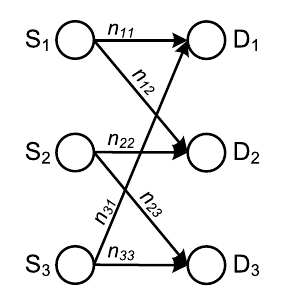}}
\hspace{0.3in}
\subfigure[]{\includegraphics[height = 3cm]{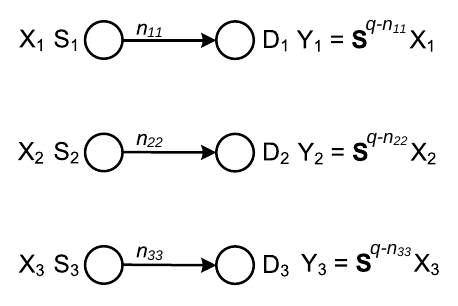}}
\caption{\it (a) A network in which coding is required to achieve normalized sum-capacity, and (b) the induced subgraphs.\label{fig:33}}
\end{figure}

Consider the linear deterministic model. By using repetition coding at the sources (as in \cite{VaneetIT}), we show that it is possible to achieve $\alpha=\frac{1}{2}$. Consider the induced subgraphs of all three S-D pairs, as shown in Figures \ref{fig:33}(b). We show that any transmission strategy over these three induced subgraphs can be implemented in the original network by using only two time-slots, such that all nodes receive the same signal as if they were in the induced subgraphs. This would immediately imply that a normalized sum-rate of $\frac{1}{2}$ is achievable.\footnote{Since any transmission strategy for the diamond networks can be implemented in the original network by using only two time-slots, we can implement the strategies that achieve the capacity for any S-D pair $i$ with full network knowledge, \emph{i.e.} $C_i$, over two time-slots. Hence, we can achieve $\frac{1}{2} \left( C_1 + C_2 + C_3 \right)$. On the other hand, we have $C_{\mathrm{sum}} \leq C_1 + C_2 + C_3$. As a result, we can achieve a set of rates such that $\sum_{i=1}^3{R_i} \geq \frac{1}{2} C_{\mathrm{sum}}$, and by the definition of normalized sum-rate, we achieve $\alpha = \frac{1}{2}$.}

To achieve $\alpha = \frac{1}{2}$, we split the communication block into two time-slots of equal length and represent each time-slot with a color, namely black and white. Sources $1$ and $2$ transmit the same codewords as if they are in the induced subgraphs over time-slot black. Destination ${\sf D}_1$ will receive the same signal as if it is only in the induced subgraph without any interference and destination ${\sf D}_3$ receives interference from source ${\sf S}_2$. Over time-slot white, source ${\sf S}_3$ transmits the same codewords as if they are in the induced subgraphs, and source ${\sf S}_2$ repeats its transmitted signal from time-slot black. Destination ${\sf D}_2$ will receive its signal interference-free. Now, if destination ${\sf D}_3$ adds its received signals over two time-slots, it recovers its intended signal interference-free, see Figure~\ref{fig:33-ach}. In other words, we have used interference cancellation at destination ${\sf D}_3$. Therefore, all S-D pairs can effectively communicate interference-free over two time-slots.

\begin{figure}[ht]
\centering
\includegraphics[height = 5cm]{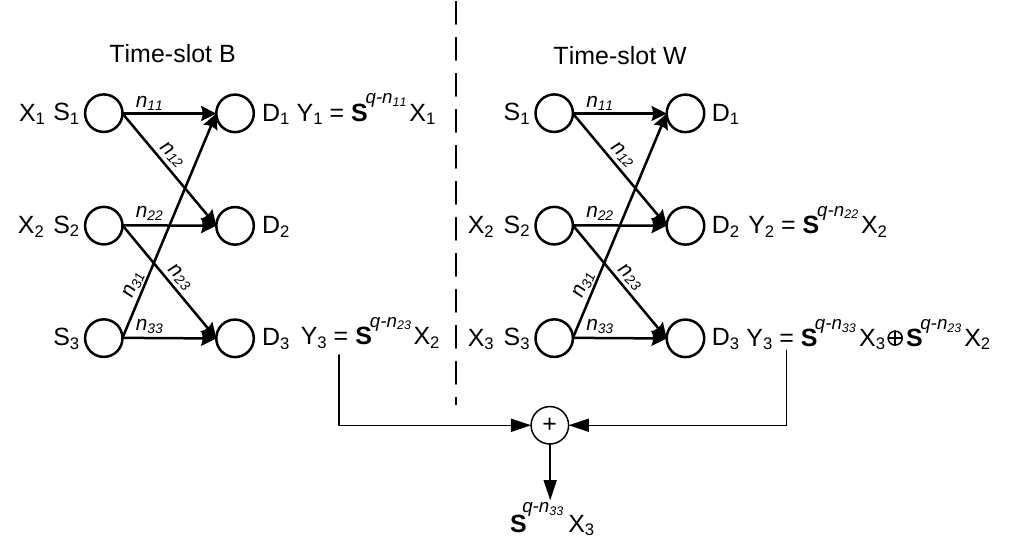}
\caption{\it Achievability strategy for the network depicted in Figure~\ref{fig:33}.\label{fig:33-ach}}
\end{figure}

Again, we can view this strategy as a modification of the previous colorings of the nodes in the route-expanded graph as follows. Each shape, \emph{i.e.} pair ID, can be assigned a subset of colors such that any two nodes with different shapes that are connected, have either different colors, or if they share a color, one them has a different color in its subset. Figure~\ref{fig:33-exp} illustrates such coloring of nodes in this example by using only two colors, $B$ and $W$. The subset $\{B,W\}$ assigned to source $2$ represents \emph{repetition coding}, \emph{i.e.} the transmitted signal in time-slot white is the same as the one transmitted in time-slot black. Since the interference can be cancelled out as described before, each S-D pair has a chance to communicate over its induced subgraph interference-free during half of the communication block.

\begin{figure}[ht]
\centering
\includegraphics[height = 3cm]{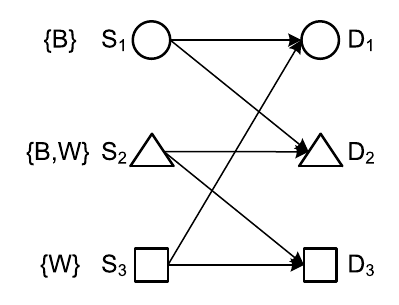}
\caption{\it Route-expanded graph of the network depicted in Figure~\ref{fig:33} and a coloring that yields $\alpha = \frac{1}{2}$.\label{fig:33-exp}}
\end{figure}

This example illustrated that with only 1-local view it is still possible to take advantage of (repetition) coding at the sources and go beyond interference avoidance. This raises a natural question: can we also exploit network coding at the relays with only $1$-local view? If so, what is a systematic procedure for doing that?

To shed light on the aforementioned questions, consider a multi-layer network as depicted in Figure~\ref{fig:333}(a). Assume linear deterministic model for the channels. It is straightforward to see that by using interference avoidance, we can at-most achieve normalized sum-rate of $\alpha=\frac{1}{3}$. We now show that by using repetition coding at the sources and linear coding at the relays, it is possible to achieve $\alpha=\frac{1}{2}$. Consider the induced subgraphs of all three S-D pairs, as shown in Figures \ref{fig:333}(b), \ref{fig:333}(c), and \ref{fig:333}(d). We now show that any transmission strategy over these three induced subgraphs can be implemented in the original network by using only two time-slots, such that all nodes receive the same signal as if they were in the diamond network. Therefore, a normalized sum-rate of $\frac{1}{2}$ is achievable.

\begin{figure}[ht]
\centering
\subfigure[]{\includegraphics[height = 2.5cm]{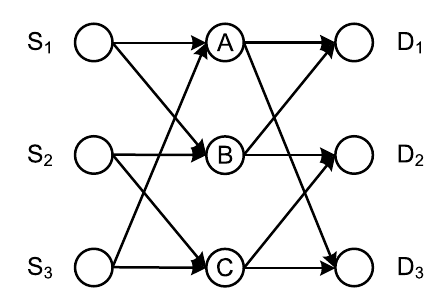}}
\hspace{-0.1in}
\subfigure[]{\includegraphics[width = 4.3cm]{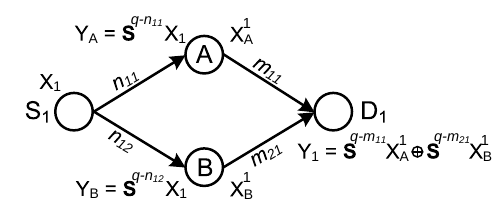}}

\subfigure[]{\includegraphics[width = 4.3cm]{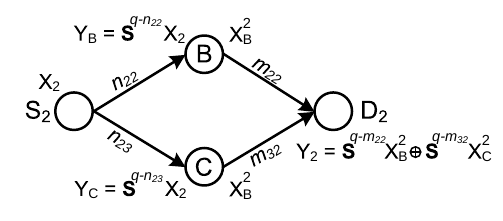}}
\hspace{-0.1in}
\subfigure[]{\includegraphics[width = 4.3cm]{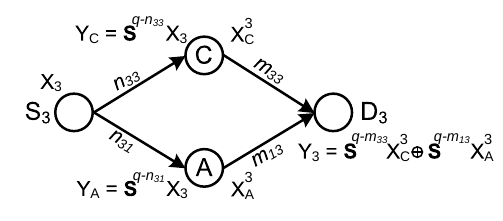}}
\caption{\it (a) a two-layer network in which we need to incorporate network coding to achieve the normalized sum-capacity, (b), (c) and (d) the induced subgraphs of S-D pairs 1,2,and 3 respectively.\label{fig:333}}
\end{figure}

Consider any strategy for S-D pairs $1$, $2$, and $3$ as illustrated in Figures \ref{fig:333}(b), \ref{fig:333}(c), and \ref{fig:333}(d). In the first layer, we implement the achievability strategy of Figure~\ref{fig:33-ach} and we have illustrated it in Figure~\ref{fig:333-ach1}. As it can be seen in this figure, at the end of the second time-slot, each relay has access to the same received signal as if it was in the diamond networks of Figures \ref{fig:333}(b), \ref{fig:333}(c), and \ref{fig:333}(d).

\begin{figure}[ht]
\centering
\includegraphics[height = 5cm]{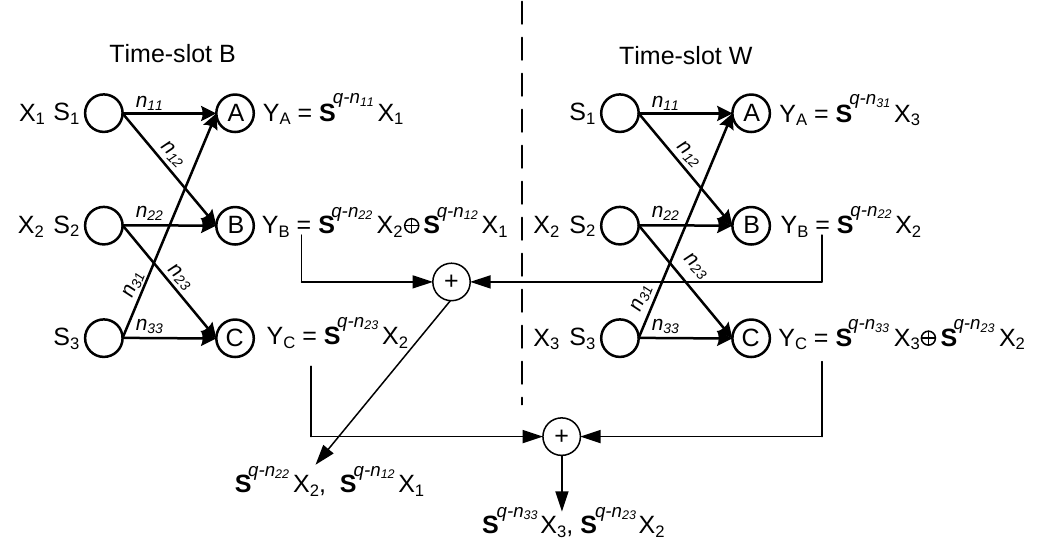}
\caption{\it Achievability strategy for the first layer of the network in Figure~\ref{fig:333}(a).\label{fig:333-ach1}}
\end{figure}

In the second layer, during time-slot black, relays $\sf A$ and $\sf B$ transmit $X_A^1$ and $X_B^1$ respectively, whereas, relay $\sf C$ transmits $X_C^2 \oplus X_C^3$, see Figure~\ref{fig:333-ach2}. Destination ${\sf D}_1$ receives the same signal as in Figure~\ref{fig:333}(b). During time-slot white, relays $\sf B$ and $\sf C$ transmit $X_B^2$ and $X_C^2$ respectively, whereas, relay $\sf A$ transmits $X_A^1 \oplus X_A^3$. Destination ${\sf D}_2$ receives the same signal as as in Figure~\ref{fig:333}(c). If destination ${\sf D}_3$ adds its received signals over the two time-slots, it recovers the same signal as in Figure~\ref{fig:333}(d). Therefore, each destination receives the same signal as if it was only in its corresponding diamond network, over two time-slots. Hence, the normalized sum-rate of $\alpha=\frac{1}{2}$ is achievable. By Lemma~\ref{lemma:generalhalf}, we know that $\alpha=\frac{1}{2}$ is also an upper-bound on the normalized sum-rate of this network, hence, we have achieved it normalized sum-capacity.

\begin{figure}[ht]
\centering
\includegraphics[height = 5cm]{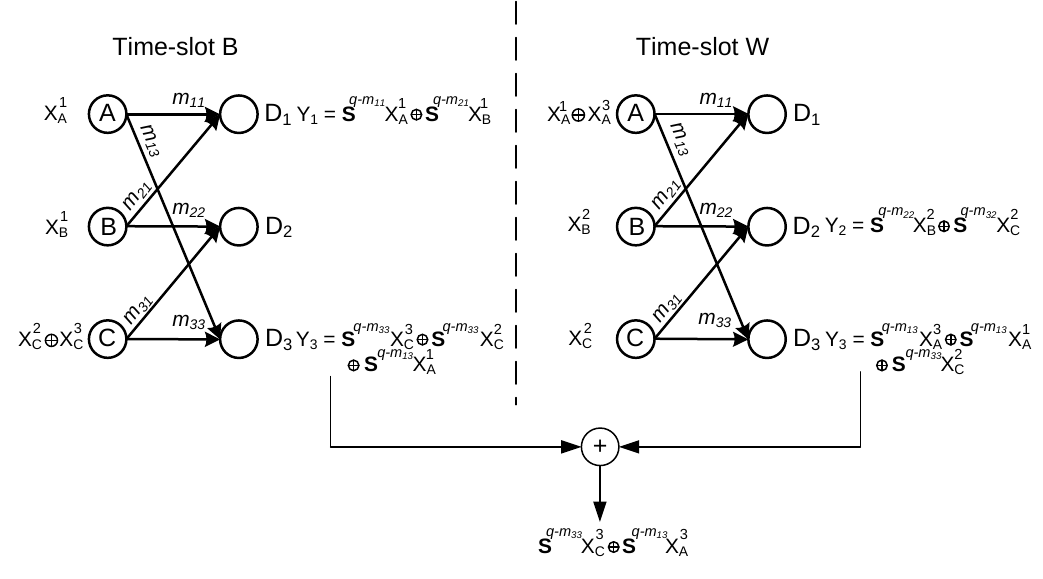}
\caption{\it Achievability strategy for the second layer of the network in Figure~\ref{fig:333}(a).\label{fig:333-ach2}}
\end{figure}

\begin{figure}[ht]
\centering
\includegraphics[height = 5cm]{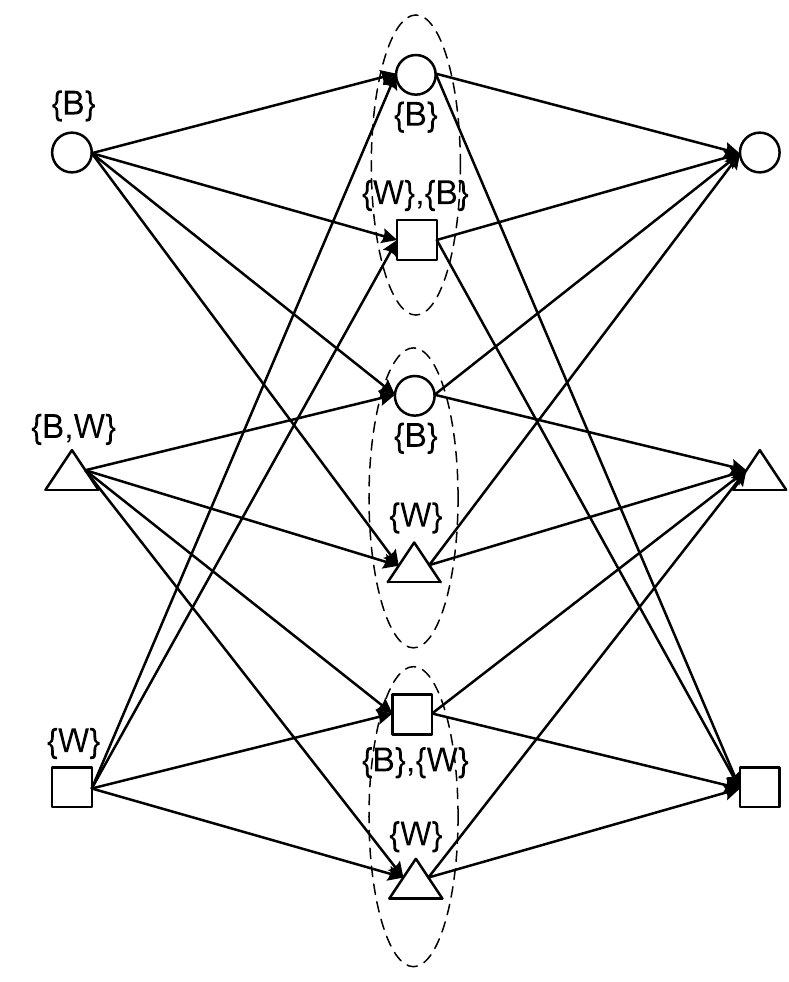}
\caption{\it The route-expanded graph for the the two-layer $(3,2)$ folded-chain network.\label{fig:MCLexpanded}}
\end{figure}

This strategy can be viewed as a new coloring of the nodes in the route-expanded graph as follows. Each shape, \emph{i.e.} pair ID, can be assigned two subsets of colors. Figure~\ref{fig:MCLexpanded} illustrates such coloring of nodes in this example by using only two colors, $B$ and $W$. The subset $\{B,W\}$ assigned to source $2$ represents repetition coding as before. The second subset of colors can be interpreted as the time instants from which we can add (or subtract for the Gaussian model) the codewords to perform network coding. To clarify, consider the first super-relay in Figure~\ref{fig:MCLexpanded}, node circle communicates the codeword of S-D pair $1$ over time-slot $B$. Over time-slot $W$, node square within this super-relay adds the codeword of S-D pair $1$ transmitted by the other node in the same super-node over time-slot $B$ to the codeword of S-D pair $3$ it has to send. Similar interpretation can be used for the other tuple of colors in this route-expanded graph. Since the interference can be cancelled out as described before, each S-D pair has a chance to communicate over its induced subgraph interference-free during half of the communication block.

In the following section, we incorporate all the ideas developed for the examples in this section to define a transmission strategy, \emph{i.e.} coded layer scheduling, which outperforms interference avoidance techniques with $1$-local view. We also characterize its performance and later in Section~\ref{performance}, we evaluate its performance for some network topologies. 

\section{Coded Layer Scheduling}
\label{section:CL}

Via the examples presented in Section~\ref{section:examples}, we saw that multiple ideas can be incorporated to enhance the achievablility scheme in multi-layer networks with $1$-local view: (1) per layer interference avoidance, (2) repetition coding to allow overhearing of the interference, and (3) network coding to allow interference neutralization. In this section, we define a general transmission strategy, named coded layer scheduling to incorporate all the aforementioned ideas. This scheduling can be represented by a specific coloring of nodes in the route-expanded graph (defined in Section~\ref{problem}). We refer to this coloring as the Coded Layer coloring and it is defined as follows.

Consider a multi-layer wireless network $\mathcal{G} = (\mathcal{V},\mathcal{E},\{ w_{ij} \}_{(i,j) \in \mathcal{E}})$, and its corresponding route-expanded graph $\mathcal{G}_{\sf exp} = (\mathcal{V}_{\sf exp},\mathcal{E}_{\sf exp})$. A \emph{Coded Layer coloring} of $\mathcal{G}_{\sf exp}$ with $T$ distinct colors $\{ c_0,c_1, \ldots,c_{T-1} \}$ assigns to any node ${\sf V}_{i,j} \in \mathcal{V}_{\sf exp}$,

\begin{enumerate}
\item a {\it transmit color set}, denoted by $\mathcal{T}_{i,j} \subseteq \mathcal{C}$, which represents the time instants in which ${\sf V}_i$ will be transmitting for S-D pair $j$ using repetition coding,
\item a {\it coding color set}, denoted by $\mathcal{C}_{i,j} \subseteq \mathcal{C}$, which represents the time instants from which node ${\sf V}_i$ will use the transmit signal to perform network coding for S-D pair $j$,
\item a {\it receive color set}, denoted by $\mathcal{R}_{i,j} \subseteq \mathcal{C}$, which represents the time instants in which it is listening.\footnote{We refer to the transmit color set, the coding color set, and the receive color set of source ${\sf S}_i$, $i=1,\ldots,K$, by $\mathcal{T}_{{\sf S}_i}$, $\mathcal{C}_{{\sf S}_i}$, and $\mathcal{R}_{{\sf S}_i}$ respectively, similar notations hold for destinations, \emph{i.e.} $\mathcal{T}_{{\sf D}_i}$, $\mathcal{C}_{{\sf D}_i}$, and $\mathcal{R}_{{\sf D}_i}$ for destination ${\sf D}_i$, $i=1,\ldots,K$.}
\end{enumerate}

To describe the conditions that these color assignments should satisfy, we need a few definitions.

\begin{definition}
\label{def:interferer}
At any node ${\sf V}_{i,j} \in \mathcal{V}_{\sf exp}$, a node ${\sf V}_{i^\prime,j^\prime} \in  \mathcal{V}_{\sf exp}$ is called an \emph{interferer} if $(i^\prime,i) \in \mathcal{V}$ and
\begin{enumerate}

\item $j^\prime \neq j$, \emph{i.e.} an interferer should have a S-D pair ID different from $j$,

\item $\nexists \hspace{1mm} {\sf V}_{i^\prime,j}: \mathcal{T}_{i^\prime,j^\prime} \cap \mathcal{C}_{i^\prime,j} \neq \emptyset$, \emph{i.e.} the colors used by an interferer are not used by any node in the same super-node that has S-D pair ID $j$ and performs network coding, otherwise its transmit signal will be neutralized.

\item $\mathcal{T}_{i^\prime,j^\prime} \cap \underbrace{ \left( \mathcal{R}_{i,j} \cap \left[ \cup_{i^\prime: (i^\prime,i) \in \mathcal{E}}{\left( \mathcal{T}_{i^\prime,j} \cup \mathcal{C}_{i^\prime,j} \right)} \right] \right) }_{\tilde{\mathcal{R}}_{i,j}} \neq \emptyset$, \emph{i.e.} an interferer transmits during a time instant that some node with S-D pair ID $j$ is transmitting to ${\sf V}_{i,j}$ and ${\sf V}_{i,j}$ is listening, the set of all such time instants is denoted by $\tilde{\mathcal{R}}_{i,j}$.

\end{enumerate}
\end{definition}

\begin{definition}
We define $\mathcal{N}_{i,j}$ as the set of all nodes in $\mathcal{V}_{\sf exp}$ that have pair ID $j$ and are connected to ${\sf V}_{i,j}$, \emph{i.e.}
\begin{equation}
\mathcal{N}_{i,j} = \{ {\sf V}_{i^\prime,j} \in \mathcal{V}_{\sf exp} | (i^\prime,i) \in \mathcal{E} \}.
\end{equation}
\end{definition}



The conditions on the assignment of $\mathcal{T}_{i,j}$, $\mathcal{C}_{i,j}$ and $\mathcal{R}_{i,j}$ of a coded layer coloring are as follows.

\textbf{C.1}: The transmit color sets assigned to the nodes that belong to the same super-node are disjoint, \emph{i.e.}
\begin{equation} 
\mathcal{T}_{i,j} \cap \mathcal{T}_{i,j^\prime} = \emptyset, \quad \forall j^\prime \neq j, i = 1,2, \ldots,|\mathcal{V}|.
\end{equation}

\textbf{C.2}: If a node is performing network coding, it only transmits once, \emph{i.e.} if $\mathcal{C}_{i,j} \neq \emptyset$, then $|\mathcal{T}_{i,j}| = 1$.

\textbf{C.3}: The coding color set $\mathcal{C}_{i,j}$ includes at most one color from each transmit color set of a node within the same super-node who is not performing network coding, \emph{i.e.}
\begin{align}
& \forall ~c_{k_1}, c_{k_2} \in \mathcal{C}_{i,j}, ~c_{k_1} \neq c_{k_2}, \exists ~j_1 \neq j_2: \nonumber \\
& c_{k_1} \in \mathcal{T}_{i,j_1}, c_{k_2} \in \mathcal{T}_{i,j_2}, \mathcal{C}_{i,j_1} = \mathcal{C}_{i,j_2} = \emptyset
\end{align}

\textbf{C.4}:  The receive color set $\mathcal{R}_{i,j}$ includes at least one color from each $\mathcal{T}_{i^\prime,j}$ such that $(i^\prime,i) \in \mathcal{E}$, \emph{i.e.}
\begin{equation}
\mathcal{R}_{i,j} \cap \mathcal{T}_{i^\prime,j} \neq \emptyset: \quad \forall ~i^\prime: (i^\prime,i) \in \mathcal{E}, j \in \{1,\ldots,K \}. 
\end{equation} 

\textbf{C.5}:  The receive color set $\mathcal{R}_{i,j}$ includes each $\mathcal{C}_{i^\prime,j}$ such that $(i^\prime,i) \in \mathcal{E}$, \emph{i.e.}
\begin{equation}
\mathcal{C}_{i^\prime,j} \subseteq \mathcal{R}_{i,j} \neq \emptyset: \quad \forall ~i^\prime: (i^\prime,i) \in \mathcal{E}, j \in \{1,\ldots,K \}. 
\end{equation}

\textbf{C.6}: If ${\sf V}_{i^\prime,j} \in \mathcal{N}_{i,j}$, then $|\mathcal{T}_{i^\prime,j} \cap \mathcal{R}_{i,j}| = 1$. 

\textbf{C.7}: At each node ${\sf V}_{i,j} \in \mathcal{V}_{\sf exp}$ either there are no interferers, or all interferers share a common color in their transmit color sets, which is in $\mathcal{R}_{i,j} \setminus \cup_{i^\prime: (i^\prime,i) \in \mathcal{E}}{\left( \mathcal{T}_{i^\prime,j} \cup \mathcal{C}_{i^\prime,j} \right)}$, \emph{i.e.}
\begin{align}
\label{eq:share}
& \left|  \left( \bigcap_{{\sf V}_{i^\prime,j}: \text{ an interferer at } {\sf V}_{i,j}}{\mathcal{T}_{i^\prime,j}} \right) \bigcap \right. \nonumber \\
& \left. \left( \mathcal{R}_{i,j} \setminus \bigcup_{i^\prime: (i^\prime,i) \in \mathcal{E}}{\left( \mathcal{T}_{i^\prime,j} \cup \mathcal{C}_{i^\prime,j} \right)} \right) \right| = 1
\end{align}
moreover, the color that the interferers share should be exclusive to them, \emph{i.e.,} for $j \neq j$,
\begin{align}
c_{i,j}^\ast \notin \bigcup_{{\sf V}_{i^\prime,j^\prime}: \text{ not an interferer at } {\sf V}_{i,j}, \left( i^\prime, i \right) \in \mathcal{E}}{\left( \mathcal{T}_{i^\prime,j^\prime} \cup \mathcal{C}_{i^\prime,j^\prime} \right)}.
\end{align}

Based on the coded layer coloring of nodes in $\mathcal{G}_{\sf exp}$, we now define the coded layer scheduling of nodes in $\mathcal{G}$ as follows:

The transmission is broken into $N$ blocks of size $T$ time instants. At the beginning of the $\ell^{th}$ block, $\ell = 1,\ldots,N$,

\begin{itemize}

\item Source ${\sf S}_i$, $i=1,\ldots,K$, creates a signal $U_{{\sf S}_i}(\ell)$, which is a function of its message $\hbox{W}_i$ ($U_{{\sf S}_i}(\ell) \in \mathbb{F}^q_2$ for the linear deterministic model and $U_{{\sf S}_i}(\ell) \in \mathbb{C}$ for the Gaussian model such that it satisfies the average power constraint at transmit nodes). The choice of this function depends on the specific strategy that each source picks,

\item Each relay node ${\sf V}_i$ creates a signal $U_{{\sf V}_{i,j}}(\ell)$ for each S-D pair $j: j \in \mathcal{J}_{{\sf V}_i}$, which is a function of its received signals $\{ Y_{{\sf V}_i}[mT+r]: 0 \leq m \leq \ell-2, c_r \in \mathcal{R}_{i,j} \}$ and the global side information. The choice of this function depends on the specific strategy that each relay picks.

\end{itemize}

During the $\ell^{th}$ block,

\begin{itemize}

\item Source ${\sf S}_i$, $i=1,\ldots,K$, will transmit 
\begin{equation}
\label{eq:CLStransmit-source}
X_{{\sf S}_i}[(\ell-1)T+t] =
\left\{ \begin{array}{ll}
U_{{\sf S}_i}(\ell) & \text{if } c_t \in \mathcal{T}_{{\sf S}_i} \\
0 & \text{otherwise}
\end{array} \right.
\end{equation}
where $t=0,\ldots, T-1$.

\item Each relay node ${\sf V}_i$ will transmit 
\begin{equation}
\label{eq:CLStransmit-relay}
X_{{\sf V}_i}[(\ell-1)T+t] =
\left\{ \begin{array}{ll}
U_{{\sf V}_{i,j}}(\ell) \\
\quad \text{if } c_t \in \mathcal{T}_{i,j} \text{ and } \mathcal{C}_{i,j} = \emptyset \\
U_{{\sf V}_{i,j}}(\ell) - \sum_{j^\prime \in \mathcal{J}_{{\sf V}_i}, j^\prime \neq j}{U_{{\sf V}_{i,j^\prime}}(\ell)} \\
\quad \text{if } c_t \in \mathcal{T}_{i,j} \text{ and } \mathcal{C}_{i,j} \neq \emptyset \\
0 \qquad \qquad \text{otherwise}
\end{array} \right.
\end{equation}
where $t=0,\ldots, T-1$. Note that subtraction in $\mathbb{F}_2^q$ is the same as XOR operation.

\end{itemize}

Finally, each destination ${\sf D}_i$, $i=1,\ldots,K$, will decode $\hbox{W}_i$ based on its received signals $\{ Y_{{\sf D}_i}[mT+r]: 0 \leq m \leq N-1, c_r \in \mathcal{R}_{{\sf D}_i} \}$ and the global side information.

We next state our main result for the coded layer scheduling.

\begin{theorem}
\label{THM:CL}
For a multi-layer network (linear deterministic or Gaussian) $\mathcal{G} = (\mathcal{V},\mathcal{E},\{ w_{ij} \}_{(i,j) \in \mathcal{E}})$ with $1$-local view, if there exists a coded layer coloring of $\mathcal{G}_{\sf exp}$ with $T$ colors as defined above, then a normalized sum-rate of $\alpha = \frac{1}{T}$ is achievable by coded layer scheduling.
\end{theorem}

\begin{proof}


We first prove the theorem for the linear deterministic model. Assume that there exists a coded layer coloring of nodes ${\sf V}_{i,j} \in \mathcal{V}_{exp}$ with colors $ \{ c_0,c_1,\ldots,c_{T-1} \}$, denoted by $\mathcal{T}_{i,j}$, $\mathcal{C}_{i,j}$ and $\mathcal{R}_{i,j}$. Suppose $\mathcal{G}$ has $K$ S-D pairs and consider the induced subgraphs of all S-D pairs, \emph{i.e.} $\mathcal{G}_{jj}$, $j=1,2,\ldots,K$. We will show that by using the coded layer scheduling, any transmission snapshot over these induced subgraphs can be implemented in the original network $\mathcal{G}$ over $T$ time instants, such that all nodes receive the same signal as if they were in the induced subgraphs. 

Consider a transmission snapshot in the $K$ induced subgraphs where
\begin{itemize}
\item Node ${\sf V}_i$ in the induced subgraph $\mathcal{G}_{j,j}$ transmits $X^j_{V_i}$,
\item Node ${\sf V}_i$ in the induced subgraph $\mathcal{G}_{j,j}$ receives
\begin{align}
\label{eq:receive-j}
Y^j_{{\sf V}_i} = \sum_{i^\prime : (i^\prime,i) \in \mathcal{E}}{{\bf S}^{q-n_{i^\prime i}} X^j_{{\sf V}_{i^\prime}}}.
\end{align}
\end{itemize}

\noindent \textbf{Transmission strategy}: At any time instant $t=0,1,\ldots,T-1$, node ${\sf V}_i \in \mathcal{V}$ will choose $U_{{\sf V}_{i,j}} = X^j_{{\sf V}_i}$ and will transmit
\begin{equation}
\label{eq:transmit}
X_{{\sf V}_i}[t] =
\left\{ \begin{array}{ll}
X^j_{{\sf V}_i} \oplus \sum_{k : \mathcal{T}_{i,k} \cap \mathcal{C}_{i,j} \neq \emptyset}{X^k_{{\sf V}_i}} & \text{if }c_t \in \mathcal{T}_{i,j}, \text{ and}\\
& j \in \{ 1,2,\ldots,K \}, \\
0 & \text{otherwise}
\end{array} \right.
\end{equation}
and it will receive
\begin{equation}
\label{eq:receive}
Y_{{\sf V}_i}[t] = \sum_{i^\prime: (i^\prime,i) \in \mathcal{E}}{{\bf S}^{q-n_{i^\prime i}} X_{{\sf V}_{i^\prime}}[t]}, \quad t=0,\ldots,T-1
\end{equation}
where summation is carried on in $\mathbb{F}_2^q$.

\noindent \textbf{Constructing the received signals}: Based on the transmission strategy described above, we need to show that at any node ${\sf V}_i$, the received signal $Y^j_{{\sf V}_i}$ can be obtained. At any node ${\sf V}_i \in \mathcal{G}_{jj}$, $j=1,2,\ldots,K$, we create $\tilde{Y}^j_{{\sf V}_i}$ as follows,
\begin{align}
\label{eq:receive-tilde}
\tilde{Y}^j_{{\sf V}_i} = \sum_{t: c_t \in \mathcal{R}_{i,j}}{Y_{{\sf V}_i}[t]},
\end{align}
where $Y_{{\sf V}_i}[t]$ is given by (\ref{eq:receive}). We will show that $\tilde{Y}^j_{{\sf V}_i} = Y^j_{{\sf V}_i}$. We have
\begin{align}
\label{eq:simplify}
\tilde{Y}^j_{{\sf V}_i} & = \sum_{t: c_t \in \mathcal{R}_{i,j}}{Y_{{\sf V}_i}[t]} \nonumber \\
& = \sum_{t: c_t \in \mathcal{R}_{i,j}}{\sum_{i^\prime: (i^\prime,i) \in \mathcal{E}}{{\bf S}^{q-n_{i^\prime i}} X_{{\sf V}_{i^\prime}}[t]}} \nonumber \\
& = \sum_{i^\prime: (i^\prime,i) \in \mathcal{E}}{{\bf S}^{q-n_{i^\prime i}} \sum_{t: c_t \in \mathcal{R}_{i,j}}{ X_{{\sf V}_{i^\prime}}[t]}} 
\end{align}

Compairing (\ref{eq:receive-j}) and (\ref{eq:simplify}), we conclude that in order to show $\tilde{Y}^j_{{\sf V}_i} = Y^j_{{\sf V}_i}$, it is sufficient to prove
\begin{align}
\sum_{t: c_t \in \mathcal{R}_{i,j}}{ X_{{\sf V}_{i^\prime}}[t]} = X^j_{{\sf V}_{i^\prime}}.
\end{align}

Consider a node ${\sf V}_{i^\prime}$ such that $(i^\prime,i) \in \mathcal{E}$; we face 2 cases:

\noindent \underline{Case 1}: $\mathcal{C}_{i^\prime,j} = \emptyset$, then from condition \textbf{C.1} and (\ref{eq:transmit}), we get
\begin{equation}
\label{eq:repetition1}
X_{{\sf V}_i}[t] = X^j_{{\sf V}_{i^\prime}} \qquad \forall ~t \text{ s.t. } c_t \in \mathcal{T}_{i^\prime,j}.
\end{equation}
Then based on condition \textbf{C.7}, we have 2 sub-cases:

\underline{Case 1-a}: There are no interferers. In this case since no interfering signal is received during time instants associated with colors in $\mathcal{R}_{i,j}$, hence,
\begin{equation}
\label{eq:repetition2}
X_{{\sf V}_i}[t] = 0 \quad \forall ~t \text{ s.t. } c_t \in \left( \mathcal{R}_{i,j} \setminus \mathcal{T}_{i^\prime,j} \right). 
\end{equation}
Then, we have
\begin{align}
\label{eq:Case1-a}
\sum_{t: c_t \in \mathcal{R}_{i,j}}{X_{{\sf V}_{i^\prime}}[t]} \overset{(\ref{eq:repetition2})}{=}  \sum_{t: c_t \in \left( \mathcal{R}_{i,j} \cap \mathcal{T}_{i^\prime,j} \right)}{X_{{\sf V}_{i^\prime}}[t]} \overset{\textbf{C.6}, (\ref{eq:repetition1})}{=} X^j_{{\sf V}_{i^\prime}}.
\end{align}

\underline{Case 1-b}: All interferers share a common color that is in $\mathcal{R}_{i,j}$ but not in $\cup_{i^\prime: (i^\prime,i) \in \mathcal{E}}{\left( \mathcal{T}_{i^\prime,j} \cup \mathcal{C}_{i^\prime,j} \right)}$. In this case, the transmit signal of any interferer ${\sf V}_{i^\prime,j^\prime}$ appears exactly twice in time instants associated with colors in $\mathcal{R}_{i,j}$, once during the time instant associated with the color that all interferers share, see (\ref{eq:share}), and once during a different time instant, see Definition~\ref{def:interferer}. Hence, when adding the received signals over all time instants associated with colors in $\mathcal{R}_{i,j}$, the transmit signal of any interferer gets canceled. Moreover, condition \textbf{C.6} guarantees that the desired signal, \emph{i.e.} $X^j_{{\sf V}_{i^\prime}}$, is transmitted exactly once during time instants associated with colors in $\mathcal{R}_{i,j}$. Hence, from (\ref{eq:repetition1}) and the argument presented above, we have
\begin{align}
\label{}
\sum_{t: c_t \in \mathcal{R}_{i,j}}{X_{{\sf V}_{i^\prime}}[t]} = X^j_{{\sf V}_{i^\prime}}.
\end{align}

\noindent \underline{Case 2}: $\mathcal{C}_{i^\prime,j} \neq \emptyset$. Note that according to condition \textbf{C.3} the colors in $\mathcal{C}_{i^\prime,j}$ can only appear in $\mathcal{T}_{i^\prime,j^\prime}$ where $\mathcal{C}_{i^\prime,j^\prime} = \emptyset$. As a result, each one of those codewords added to $X^j_{{\sf V}_{i^\prime}}$ as in (\ref{eq:transmit}), are also transmitted during time instants corresponding to colors in $\mathcal{C}_{i,j}$. Hence, if $\mathcal{C}_{i^\prime,j} \neq \emptyset$,
\begin{align}
\label{eq:coding1}
\sum_{t: c_t \in \mathcal{T}_{i^\prime,j} \cup \mathcal{C}_{i^\prime,j}}{X_{{\sf V}_{i^\prime}}[t]} =  X^j_{{\sf V}_{i^\prime}}.
\end{align}
Again based on condition \textbf{C.7}, we have 2 sub-cases:

\underline{Case 2-a}: There are no interferers. In this case since no interfering signal is received during time instants associated with colors in $\mathcal{R}_{i,j}$, hence,
\begin{equation}
\label{eq:coding2}
X_{{\sf V}_i}[t] = 0 \quad \forall ~t \text{ s.t. } c_t \in \left[ \mathcal{R}_{i,j} \setminus \left( \mathcal{T}_{i^\prime,j} \cup \mathcal{C}_{i^\prime,j} \right) \right]. 
\end{equation}
Then, we have
\begin{align}
\label{}
\sum_{t: c_t \in \mathcal{R}_{i,j}}{X_{{\sf V}_{i^\prime}}[t]} & \hspace{2mm} \overset{(\ref{eq:coding2})}{=}  \sum_{ t: c_t \in \left[ \mathcal{R}_{i,j} \cap \left( \mathcal{T}_{i^\prime,j} \cup \mathcal{C}_{i^\prime,j} \right) \right]}{X_{{\sf V}_{i^\prime}}[t]} \\
& \overset{\textbf{C.6}, (\ref{eq:coding1})}{=} X^j_{{\sf V}_{i^\prime}}.
\end{align}


\underline{Case 2-b}: All interferers share a common color that is in $\mathcal{R}_{i,j}$ but not in $\cup_{i^\prime: (i^\prime,i) \in \mathcal{E}}{\left( \mathcal{T}_{i^\prime,j} \cup \mathcal{C}_{i^\prime,j} \right)}$. The argument is similar to that of case 1-b, \emph{i.e.} the transmit signal of any interferer ${\sf V}_{i^\prime,j^\prime}$ appears exactly twice in time instants associated with colors in $\mathcal{R}_{i,j}$, once during the time instant associated with the color that all interferers share, see (\ref{eq:share}), and once during a different time slot, see Definition~\ref{def:interferer}. Hence, when adding the received signals over all time instants associated with colors in $\mathcal{R}_{i,j}$, the transmit signal of any interferer gets canceled. Moreover, condition \textbf{C.6} guarantees that the desired signal, \emph{i.e.} $X^j_{{\sf V}_{i^\prime}}$, is transmitted exactly once during time instants associated with colors in $\mathcal{R}_{i,j}$. Hence, from (\ref{eq:coding1}) and the argument presented above, we have
\begin{align}
\label{}
\sum_{t: c_t \in \mathcal{R}_{i,j}}{X_{{\sf V}_{i^\prime}}[t]} = X^j_{{\sf V}_{i^\prime}}.
\end{align}

Therefore, we have shown that in all cases we have $\sum_{t: c_t \in \mathcal{R}_{i,j}}{ X_{{\sf V}_{i^\prime}}[t]} = X^j_{{\sf V}_{i^\prime}}$, which as described before proves that $\tilde{Y}^j_{{\sf V}_i} = Y^j_{{\sf V}_i}$. As a result, at any node ${\sf V}_i \in \mathcal{G}_{jj}$, the received signal $Y^j_{{\sf V}_i}$ can be obtained interference-free over $T$ time instants. This implies that the transmit signal $X^j_{{\sf V}_i}$, which is only a function of $Y^j_{{\sf V}_i}$ and the global side information, can be created over $T$ time instants as well. Hence, any transmission snapshot over the induced subgraphs can be implemented in the original network $\mathcal{G}$ over $T$ time instants, such that all nodes receive the same signal as if they were in the induced subgraphs. 

Moreover, any transmission strategy with block length $N$ for the induced subgraphs can be implemented in the original network $\mathcal{G}$ over $N T$ time instants in a similar manner. Hence, we can implement the strategies that achieve the capacity for any S-D pair $i$ with full network knowledge, \emph{i.e.} $C_i$, as $N \rightarrow \infty$ over $N T$ time instants. Therefore, by choosing $U_{{\sf V}_{i,j}}$'s according to the optimal transmission strategies and creating $X_{{\sf V}_i}[t]$ as in (\ref{eq:transmit}), we can achieve $\frac{1}{T} \left( C_1 + C_2 + \ldots + C_K \right)$. On the other hand, we have $C_{\mathrm{sum}} \leq C_1 + C_2 + \ldots + C_K$. As a result, we can achieve a set of rates such that $\sum_{i=1}^K{R_i} \geq \frac{1}{T} C_{\mathrm{sum}}$, and by the definition of the normalized sum-rate, we achieve $\alpha = \frac{1}{T}$.

We will next prove the theorem for the Gaussian model. We will show that by using the coded layer scheduling, any transmission snapshot over the induced subgraphs $\mathcal{G}_{jj}$, $j=1,2,\ldots,K$, can be implemented in the original network $\mathcal{G}$ over $T$ time instants, such that all nodes receive the same signal as if they were in the induced subgraphs. 

Consider a transmission snapshot in the $K$ induced subgraphs where
\begin{itemize}
\item Node ${\sf V}_i$ in the induced subgraph $\mathcal{G}_{j,j}$ transmits $X^j_{V_i}$,
\item Node ${\sf V}_i$ in the induced subgraph $\mathcal{G}_{j,j}$ receives
\begin{align}
\label{eq:receive-gaussian-1}
Y^j_{{\sf V}_i} = \sum_{i^\prime : (i^\prime,i) \in \mathcal{E}}{ h_{i^\prime i} X^j_{{\sf V}_{i^\prime}}} + Z^\prime_i,
\end{align}
\end{itemize}
where $Z^\prime_i$ is the additive white complex Gaussian noise with variance $T$. We also assume a power constraint of $\frac{P}{T}$ at the transmit nodes in the induced subgraphs.

\noindent \textbf{Transmission strategy}: 
At any time instant $t=0,1,\ldots,T-1$, node ${\sf V}_i$ will choose $U_{{\sf V}_{i,j}} = X^j_{{\sf V}_i}$ and will transmit
\begin{equation}
\label{eq:transmit-gaussian}
X_{{\sf V}_i}[t] =
\left\{ \begin{array}{ll}
X^j_{{\sf V}_i} - \sum_{k : \mathcal{T}_{i,k} \cap \mathcal{C}_{i,j} \neq \emptyset}{X^k_{{\sf V}_i}} \\
\text{ if } c_t \in \mathcal{T}_{i,j}, \text{ and } j \in \{ 1,2,\ldots,K \}, \\
0 \qquad \qquad \text{otherwise}
\end{array} \right.
\end{equation}
note that number of transmit signals at each time instant is less than $T$ and due to the power constraint of $\frac{P}{T}$ in the induced subgraphs, the power constraint in the original network $\mathcal{G}$ is satisfied. At any time instant $t=1,2,\ldots,T$, node ${\sf V}_i$ will receive
\begin{equation}
\label{eq:receive-gaussian}
Y_{{\sf V}_i}[t] = \sum_{i^\prime: (i^\prime,i) \in \mathcal{E}}{ h_{i^\prime i} X_{{\sf V}_{i^\prime}}[t]} + Z_i[t],
\end{equation}
where $Z_i[t]$ is the additive white complex Gaussian noise with unit variance.

\noindent \textbf{Constructing the received signals}: Based on the transmission strategy described above, we need to show that at any node ${\sf V}_i$, $\sum_{i^\prime : (i^\prime,i) \in \mathcal{E}}{ h_{i^\prime i} X^j_{{\sf V}_{i^\prime}}} + \tilde{Z}_i$ can be obtained where $\tilde{Z}_i$ is an additive white complex Gaussian noise that has the same distribution as $Z^{\prime}_i$. At any node ${\sf V}_i \in \mathcal{G}_{jj}$, $j=1,2,\ldots,K$, we create $\tilde{Y}^j_{{\sf V}_i}$ as follows,
\begin{align}
\label{}
\tilde{Y}^j_{{\sf V}_i} = \sum_{t: c_t \in \mathcal{R}_{i,j}}{Y_{{\sf V}_i}[t]},
\end{align}
where $Y_{{\sf V}_i}[t]$ is given by (\ref{eq:receive-gaussian}). We will show that $\tilde{Y}^j_{{\sf V}_i} = \sum_{i^\prime : (i^\prime,i) \in \mathcal{E}}{ h_{i^\prime i} X^j_{{\sf V}_{i^\prime}}} + \tilde{Z}_i$, \emph{i.e.} $\tilde{Y}^j_{{\sf V}_i}$ is equal to $Y^j_{{\sf V}_i}$ defined in (\ref{eq:receive-gaussian-1}), ignoring the noise terms. We have
\begin{align}
\label{}
\tilde{Y}^j_{{\sf V}_i} & = \sum_{t: c_t \in \mathcal{R}_{i,j}}{Y_{{\sf V}_i}[t]} \nonumber \\
& = \sum_{t: c_t \in \mathcal{R}_{i,j}}{\left[ \left( \sum_{i^\prime: (i^\prime,i) \in \mathcal{E}}{ h_{i^\prime i} X_{{\sf V}_{i^\prime}}[t] } \right) + Z_i \right] } \nonumber \\
& = \left( \sum_{i^\prime: (i^\prime,i) \in \mathcal{E}}{h_{i^\prime i} \sum_{t: c_t \in \mathcal{R}_{i,j}}{  X_{{\sf V}_{i^\prime}}[t]}} \right) + \sum_{t: c_t \in \mathcal{R}_{i,j}}{Z_i[t]}.
\end{align}
Note that $\sum_{t: c_t \in \mathcal{R}_{i,j}}{Z_i[t]}$ is an additive white complex Gaussian noise with variance $|\mathcal{R}_{i,j}| \le T$. Therefore, by adding noise we can create $\tilde{Z}_i$ that has the same distribution as $Z^{\prime}_i$. Hence, it is sufficient to prove that
\begin{align}
\label{eq:toprovegaussian}
\sum_{t: c_t \in \mathcal{R}_{i,j}}{ X_{{\sf V}_{i^\prime}}[t]} = X^j_{{\sf V}_{i^\prime}}.
\end{align}


The argument to prove (\ref{eq:toprovegaussian}) is similar to that presented for the linear deterministic model with minor modifications as follows. Consider a node ${\sf V}_{i^\prime}$ such that $(i^\prime,i) \in \mathcal{E}$; we face 2 cases:

\noindent \underline{Case 1}: $\mathcal{C}_{i^\prime,j} = \emptyset$, then from condition \textbf{C.1} and (\ref{eq:transmit-gaussian}), we get
\begin{equation}
\label{eq:repetition-gaussian1}
X_{{\sf V}_i}[t] = X^j_{{\sf V}_{i^\prime}} \qquad \forall ~t \text{ s.t. } c_t \in \mathcal{T}_{i^\prime,j}.
\end{equation}
then based on condition \textbf{C.7}, we have 2 sub-cases:

\underline{Case 1-a}: There are no interferers. In this case since no interfering signal is received during time instants associated with colors in $\mathcal{R}_{i,j}$, hence,
\begin{equation}
\label{eq:repetition-gaussian2}
X_{{\sf V}_i}[t] = 0 \quad \forall ~t \text{ s.t. } c_t \in \left( \mathcal{R}_{i,j} \setminus \mathcal{T}_{i^\prime,j} \right). 
\end{equation}
Then, we have
\begin{align}
\label{eq:Case1-a}
\sum_{t: c_t \in \mathcal{R}_{i,j}}{X_{{\sf V}_{i^\prime}}[t]} \overset{(\ref{eq:repetition-gaussian2})}{=}  \sum_{t: c_t \in \left( \mathcal{R}_{i,j} \cap \mathcal{T}_{i^\prime,j} \right)}{X_{{\sf V}_{i^\prime}}[t]} \overset{\textbf{C.6}, (\ref{eq:repetition-gaussian1})}{=} X^j_{{\sf V}_{i^\prime}}.
\end{align}

\underline{Case 1-b}: All interferers share a common color that is in $\mathcal{R}_{i,j}$ but not in $\cup_{i^\prime: (i^\prime,i) \in \mathcal{E}}{\left( \mathcal{T}_{i^\prime,j} \cup \mathcal{C}_{i^\prime,j} \right)}$. In this case, the transmit signal of any interferer ${\sf V}_{i^\prime,j^\prime}$ appears exactly twice in time instants associated with colors in $\mathcal{R}_{i,j}$, once with positive sign during the time instant associated with the color that all interferers share, see (\ref{eq:share}), and once with negative sign during a different time slot, see Definition~\ref{def:interferer}  and (\ref{eq:transmit-gaussian}). Hence, when adding the received signals over all time instants associated with colors in $\mathcal{R}_{i,j}$, the transmit signal of any interferer gets canceled. Moreover, condition \textbf{C.6} guarantees that the desired signal, \emph{i.e.} $X^j_{{\sf V}_{i^\prime}}$, is transmitted exactly once during time instants associated with colors in $\mathcal{R}_{i,j}$. Hence, from (\ref{eq:repetition-gaussian1}) and the argument presented above, we have
\begin{align}
\label{}
\sum_{t: c_t \in \mathcal{R}_{i,j}}{X_{{\sf V}_{i^\prime}}[t]} = X^j_{{\sf V}_{i^\prime}}.
\end{align}

\noindent \underline{Case 2}: $\mathcal{C}_{i^\prime,j} \neq \emptyset$. Note that according to condition \textbf{C.3} the colors in $\mathcal{C}_{i^\prime,j}$ can only appear in $\mathcal{T}_{i^\prime,j^\prime}$ where $\mathcal{C}_{i^\prime,j^\prime} = \emptyset$. As a result, each one of those codewords subtracted from $X^j_{{\sf V}_{i^\prime}}$ as in (\ref{eq:transmit-gaussian}), are also transmitted during time instants corresponding to colors in $\mathcal{C}_{i,j}$ with positive sign. Hence, if $\mathcal{C}_{i^\prime,j} \neq \emptyset$,
\begin{align}
\label{eq:coding-gaussian1}
\sum_{t: c_t \in \mathcal{T}_{i^\prime,j} \cup \mathcal{C}_{i^\prime,j}}{X_{{\sf V}_{i^\prime}}[t]} =  X^j_{{\sf V}_{i^\prime}}.
\end{align}
again based on condition \textbf{C.7}, we have 2 sub-cases:

\underline{Case 2-a}: There are no interferers. In this case since no interfering signal is received during time instants associated with colors in $\mathcal{R}_{i,j}$, hence,
\begin{equation}
\label{eq:coding-gaussian2}
X_{{\sf V}_i}[t] = 0 \quad \forall ~t \text{ s.t. } c_t \in \left[ \mathcal{R}_{i,j} \setminus \left( \mathcal{T}_{i^\prime,j} \cup \mathcal{C}_{i^\prime,j} \right) \right]. 
\end{equation}
Then, we have
\begin{align}
\label{}
\sum_{t: c_t \in \mathcal{R}_{i,j}}{X_{{\sf V}_{i^\prime}}[t]} & \hspace{2mm} \overset{(\ref{eq:coding-gaussian2})}{=}  \sum_{ t: c_t \in \left[ \mathcal{R}_{i,j} \cap \left( \mathcal{T}_{i^\prime,j} \cup \mathcal{C}_{i^\prime,j} \right) \right]}{X_{{\sf V}_{i^\prime}}[t]} \\
& \overset{\textbf{C.6}, (\ref{eq:coding-gaussian1})}{=} X^j_{{\sf V}_{i^\prime}}.
\end{align}


\underline{Case 2-b}: All interferers share a common color that is in $\mathcal{R}_{i,j}$ but not in $\cup_{i^\prime: (i^\prime,i) \in \mathcal{E}}{\left( \mathcal{T}_{i^\prime,j} \cup \mathcal{C}_{i^\prime,j} \right)}$. The argument is similar to that of case 1-b, \emph{i.e.} the transmit signal of any interferer ${\sf V}_{i^\prime,j^\prime}$ appears exactly twice in time instants associated with colors in $\mathcal{R}_{i,j}$, once with positive sign during the time instant associated with the color that all interferers share, see (\ref{eq:share}), and once with negative sign during a different time slot, see Definition~\ref{def:interferer} and (\ref{eq:transmit-gaussian}). Hence, when adding the received signals over all time instants associated with colors in $\mathcal{R}_{i,j}$, the transmit signal of any interferer gets canceled. Moreover, condition \textbf{C.6} guarantees that the desired signal, \emph{i.e.} $X^j_{{\sf V}_{i^\prime}}$, is transmitted exactly once during time instants associated with colors in $\mathcal{R}_{i,j}$. Hence, from (\ref{eq:coding-gaussian1}) and the argument presented above, we have
\begin{align}
\label{}
\sum_{t: c_t \in \mathcal{R}_{i,j}}{X_{{\sf V}_{i^\prime}}[t]} = X^j_{{\sf V}_{i^\prime}}.
\end{align}

Therefore, we have shown that in all cases we have $\sum_{t: c_t \in \mathcal{R}_{i,j}}{ X_{{\sf V}_{i^\prime}}[t]} = X^j_{{\sf V}_{i^\prime}}$. As a result, effectively we have decomposed the induced dubgraphs $\mathcal{G}_{jj}$, $j=1,\ldots,K$, from the original network $\mathcal{G}$ over $T$ time instants. 

Now, in order to show that we can achieve a normalized sum-rate of $\alpha = \frac{1}{T}$, we need to prove that by changing the noise variances from $1$ to $T$ and the power constraints from $P$ to $\frac{P}{T}$, the capacity of each iduced subgraph is decreased by at most a constant that is independent of the channel gains, this has been shown in Claim~\ref{claim:kappa} in Appendix~\ref{Appendix:kappa}. Therefore, via coded layer scheduling, we achieve a sum-rate $\sum_{i=1}^K{R_i} \ge \frac{1}{T} \sum_{i=1}^K{C_i} - \tau$, where $\tau = \frac{K}{T} |\mathcal{V}| \left( 2 \log T + 17 \right)$ is a constant independent of channel gains. Hence by the definition of normalized sum-rate, we achieve $\alpha = \frac{1}{T}$.

\end{proof}

We refer to the coded layer coloring that maximizes $\frac{1}{T}$ as the {\it Maximal} Coded Layer coloring, and to the coded layer scheduling that is defined based on the Maximal Coded Layer coloring as the {\it Maximal} Coded Layer (MCL) scheduling.

\begin{remark}
A special subclass of MCL scheduling is to avoid coding (either repetition or network coding) at the nodes and just perform interference avoidance at each layer. This is simply obtained by imposing the following constraints to the coloring of nodes in $\mathcal{G}_{\sf exp}$: $$\mathcal{C}_{i,j} = \emptyset, \text{ and } |\mathcal{T}_{i,j}| = 1, \forall ~{\sf V}_{i,j} \in \mathcal{V}_{\sf exp}.$$ We refer to this coloring as the {\it Independent Layer} coloring and its corresponding scheduling as  the {\it Independent Layer} (IL) scheduling. We refer to the IL scheduling with minimum number of colors as the {\it Maximal} Independent Layer (MIL) scheduling. 

Finally, by expressing coded layer scheduling as a coloring algorithm, the special case and its relation to known schedulers in the literature becomes apparent. As explained above, MIL coloring uses only one color in coding and receive sets. Thus MIL scheduling has the feel of multi-path routing and hop-by-hop scheduling (see e.g.~\cite{LL03}) which generalizes single-path routing and scheduling~\cite{TE92,CLCD06} in multi-hop networks. However MIL scheduling differs from all routing based methods, multi-path and otherwise. In MIL scheduling, the information combining occurs at the signal level and thus without interference, each flow can be information-theoretically optimal if operating in isolation. However, in multi-path routing, even though multiple paths are used, the basic building block is a point-to-point link and thus, the smallest information block is a packet. Thus, in the absence of interference, a flow which has multiple possible routes does not achieve information-theoretic capacity since not all degrees of freedom are used at the signal level.
\end{remark}

In the following section, we evaluate the performance of MCL and MIL scheduling for some sample networks. As we will, MIL scheduling (that is based on the idea of per layer interference avoidance) is optimal for a class of networks, namely $K \times \underbrace{2 \times \ldots \times 2}_M \times K$ networks. However, MCL scheduling is optimal for a larger class of networks in which interference avoidance techniques fail to achieve the normalized sum-capacity.

\section{Optimality of the Strategies}
\label{performance}
In the previous section, we introduced the coded layer scheduling and the Independent Layer scheduling as its special case. In this section, we show the optimality of the aforementioned transmission strategies for some sample networks. We start by a class of networks where MIL scheduling is optimal. We then consider a class of networks in which coding is required to achieve the normalized sum-capacity. By characterizing the achievable normalized sum-rate via MIL scheduling, we will show that MIL scheduling is {\it not} always optimal in these networks. We then introduce a class of networks in which implementing MCL scheduling provides us with unbounded gain as opposed to MIL scheduling.

\subsection{Two-layer $K$-user networks with two relays per layer}

In this subsection, we define a class of networks where MIL scheduling is optimal but end-to-end interference avoidance is not necessarily optimal. This class of networks is an extension of those in Figures \ref{MIP} and \ref{MIL}, \emph{i.e.} two-layer $K$-user networks with two relays per layer.

\begin{definition}
A {\it $K \times \underbrace{2 \times \ldots \times 2}_M \times K$} network is a multi-layer network (as defined in Section \ref{subsec:netModel}) with $L = M + 2$, $|\mathcal{V}_1| = |\mathcal{V}_L| = K$, and $|\mathcal{V}_2| = |\mathcal{V}_3| = \ldots = |\mathcal{V}_{M+1}| = 2$. See Figure~\ref{k222k} for a depiction.
\end{definition}

\begin{definition}
A {\it non-interfering $K \times \underbrace{2 \times \ldots \times 2}_M \times K$} network is a $K \times \underbrace{2 \times \ldots \times 2}_M \times K$ network where if there exists a path from ${\sf V}^2_i$ to ${\sf V}^{M+1}_j$, then there is no path from ${\sf V}^2_{\bar{i}}$ to ${\sf V}^{M+1}_j$, $i,j \in \{ 1, 2 \}$ and $\bar{i} = 3 - i$.
\end{definition}

\begin{figure}[h]
\centering
\includegraphics[height= 3cm]{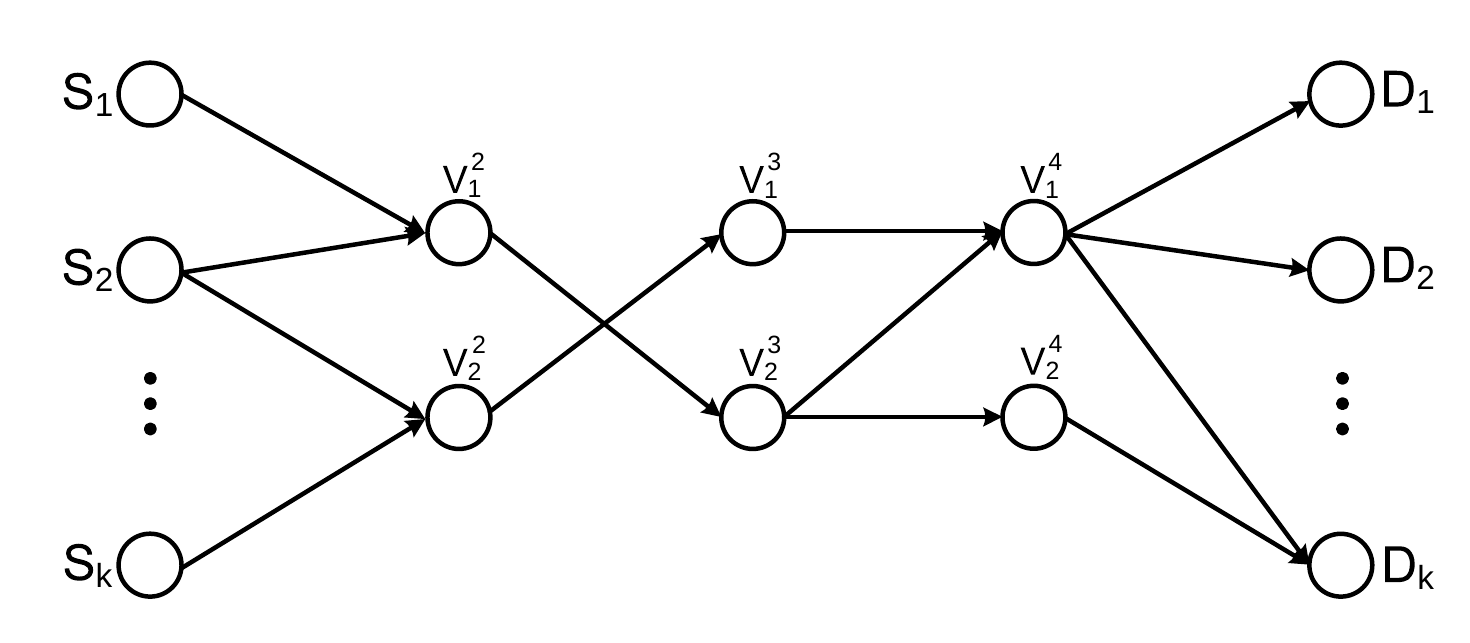}
\caption{A $K \times 2 \times 2 \times 2 \times K$ network.}
\label{k222k}
\end{figure}

\begin{theorem} \label{THM:K22K}
The normalized sum-capacity of a $K \times \underbrace{2 \times \ldots \times 2}_M \times K$ network (linear deterministic or Gaussian) with $1$-local view, is
\begin{equation}
\alpha^*=
\left\{ \begin{array}{ll}
\frac{1}{d_{\max}} & \text{if the network is non-interfering,}\\
\frac{1}{K} & \textrm{otherwise}\\
\end{array} \right.
\end{equation}
and is achieved by MIL scheduling.
\end{theorem}

\begin{proof}

\begin{figure}[h]
\centering
\includegraphics[height= 3cm]{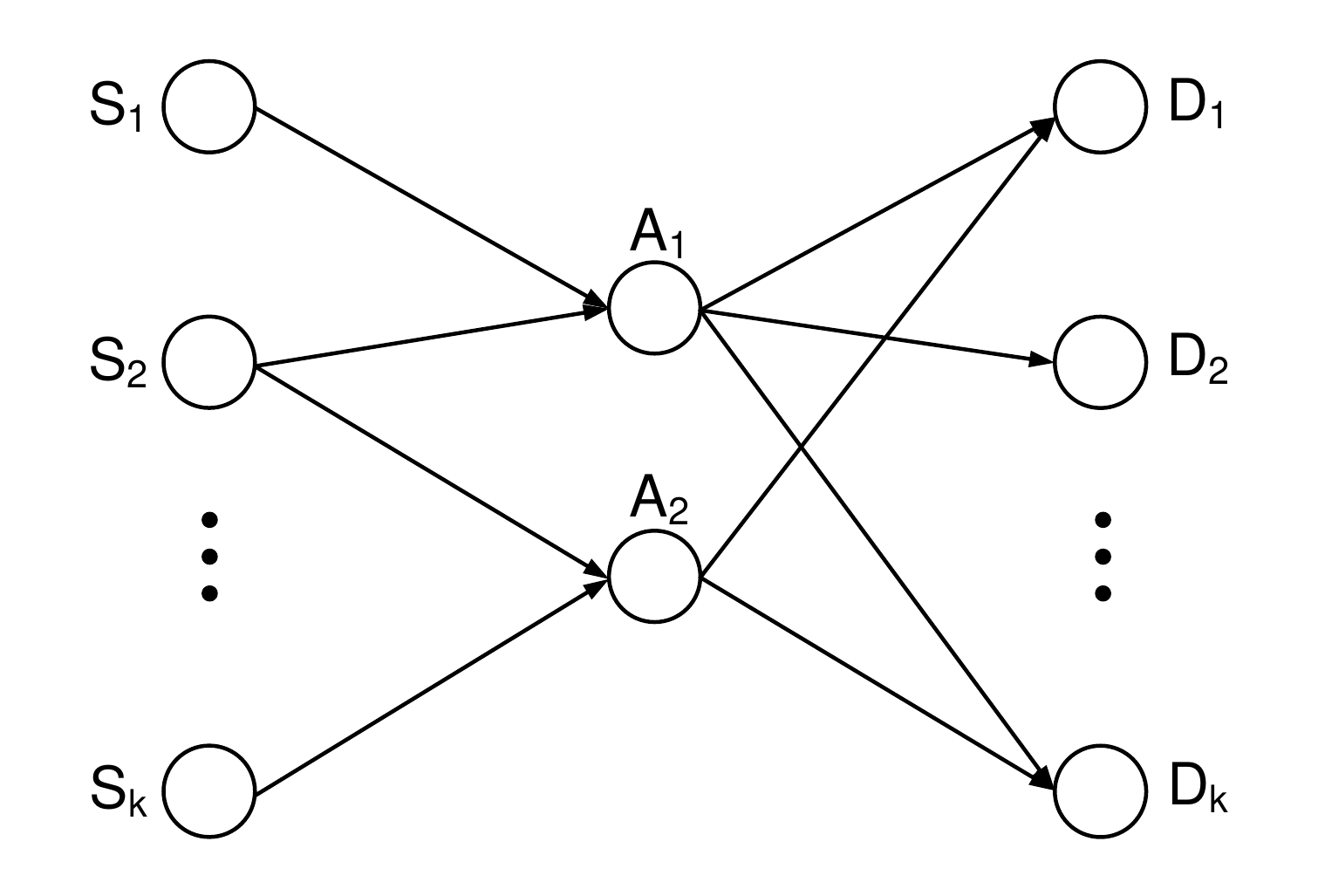}
\caption{A $K \times 2 \times K$ network.}
\label{k2k}
\end{figure}

The result for $K = 1$ is trivial. For $K \geq 2$, we will first prove the result for the special case of $K \times 2 \times K$ networks, see Figure~\ref{k2k}, and then we extend the result to $K \times \underbrace{2 \times \ldots \times 2}_M \times K$ networks.

\noindent  {\bf Converse:} Assume that a normalized sum-rate of $\alpha$ is achievable, \emph{i.e.} there exists a transmission strategy with $1$-local view, such that for all channel realizations, it achieves a sum-rate satisfying $\sum_{i=1}^K{R_i} \ge \alpha C_{\mathrm{sum}} -\tau$ with error probabilities going to zero as $N \rightarrow \infty$ and for some constant $\tau \in R$ independent of the channel gains. We will show that $\alpha \leq \frac{1}{d_{\max}}$.


Without loss of generality, it is sufficient to consider two cases: (1) $d_{\max} = d_{\mathrm{in}}({\sf A}_1)$, and (2) $d_{\max} = d_{\mathrm{out}}({\sf A}_1)$.

\noindent  {\bf Converse proof for case (1)}: We divide the S-D pair ID's into 3 disjoint subsets as follows: $\mathcal{J}_i$ is the set of all the S-D pair ID's such that the corresponding source is connected to relay ${\sf A}_i$, $i=1,2$, and $\mathcal{J}_{12}$ is the set of all the other S-D pair ID's. In other words, $\mathcal{J}_{12}$ is the set of all the S-D pair ID's where the corresponding source is connected to both relays. 

Our goal is to derive an upper bound on the normalized sum-capacity of this network by specific assignment of channel gains. Consider the corresponding destinations of the S-D pairs in $\mathcal{J}_1$. Any such destination is either connected to relay ${\sf A}_1$ or to both relays, since otherwise it cannot get its message. If it is connected to both, then set the channel gain of the link from relay ${\sf A}_2$ equal to $0$. Follow the similar steps for the members of $\mathcal{J}_2$. The corresponding destinations of the sources in $\mathcal{J}_{12}$ are either connected to only one relay or to both relays. If such a destination is connected to both relays, assign the channel gain of $0$ to one of the links connecting it to a relay (pick this link at random). Assign channel gain of $n$ (in linear deterministic model), and $h$ (in Gaussian model) to all other links in the network, where $n \in \mathbb{N}$ and $h \in \mathbb{C}$. 

Suppose each destination ${\sf D}_i$ is able to decode its message $\hbox{W}_i$, $i=1,\ldots,K$. With the channel gain assignment described above, all destinations corresponding to members of $\mathcal{J}_1$ are {\it only} connected to relay ${\sf A}_1$. Hence, relay ${\sf A}_1$ has all the information that each one of the destinations corresponding to members of $\mathcal{J}_1$ requires in order to decode its message. Therefore, relay ${\sf A}_1$ should be able to decode all the messages coming from sources corresponding to the members of $\mathcal{J}_1$. A similar claim is valid for relay ${\sf A}_2$, \emph{i.e.} it should be able to decode all the messages coming from $\mathcal{J}_2$. Therefore, relays ${\sf A}_1$ and ${\sf A}_2$ can decode the messages coming from members of $\mathcal{J}_1$ and $\mathcal{J}_2$ respectively. They decode these messages and remove them from the received signals.

Now, each relay should be able to decode the rest of the messages (in the linear deterministic case, relays have the same received signals and in the Gaussian case, they receive the same codewords with different noise terms, however, since the destinations are able to decode messages, relays should be able to do so). This means that relay ${\sf A}_1$ is able to decode all the messages from $\mathcal{J}_1$ and $\mathcal{J}_{12}$, note that $d_{\mathrm{in}}({\sf A}_1) = |\mathcal{J}_1| + |\mathcal{J}_{12}|$.

Given the assumption of $1$-local view, in order to achieve a normalized sum-rate of $\alpha$, each source should transmit at a rate greater than or equal to $\alpha n - \tau$ (in linear deterministic model) where $\tau$ is described at the beginning of the converse. This is due to the fact that from each source's point of view, it is possible that the other S-D pairs have capacity $0$, therefore in order to achieve a normalized sum-rate of $\alpha$, it should transmit at a rate of at least $\alpha n - \tau$. The MAC capacity at relay ${\sf A}_1$, gives us
\begin{equation}
d_{\mathrm{in}}({\sf A}_1) ( \alpha n - \tau ) \leq n \Rightarrow ( d_{\mathrm{in}}({\sf A}_1) \alpha - 1 ) n \leq d_{\mathrm{in}}({\sf A}_1) \tau.
\end{equation}
Since this has to hold for all values of $n$, and $\alpha$ and $\tau$ are independent of $n$, we get $\alpha \leq \frac{1}{d_{\mathrm{in}}({\sf A}_1)}$.

In the Gaussian case, each source should transmit at a rate greater than or equal to $\alpha \log(1+|h|^2) - \tau$ since from each source's point of view, it is possible that the other S-D pairs have capacity $0$. From the MAC capacity at relay ${\sf A}_1$, we get
\begin{equation}
d_{\mathrm{in}}({\sf A}_1) ( \alpha \log(1+|h|^2) - \tau ) \leq \log(1+d_{\mathrm{in}}({\sf A}_1) \times |h|^2),
\end{equation}
which results in
\begin{equation}
d_{\mathrm{in}}({\sf A}_1) ( \alpha \log(1+|h|^2) - \tau ) \leq \log(d_{\mathrm{in}}({\sf A}_1)) + \log(1 + \times |h|^2).
\end{equation}
Hence, we have
\begin{equation}
( d_{\mathrm{in}}({\sf A}_1) \alpha - 1 ) \log(1+|h|^2) \leq \log(d_{\mathrm{in}}({\sf A}_1)) + d_{\mathrm{in}}({\sf A}_1) \tau.
\end{equation}
Since this has to hold for all values of $h$, and $\alpha$ and $\tau$ are independent of $h$, we get $\alpha \leq \frac{1}{d_{\mathrm{in}}({\sf A}_1)}$.

\noindent {\bf Coverse proof for case (2)}: If a destination is only connected to relay ${\sf A}_2$, assign channel gain of $0$ to the link from ${\sf A}_2$ to such destination. Set all the other channel gains equal to $n$ (in the linear deterministic model), and equal to $h$ (in the Gaussian model), where $n \in \mathbb{N}$ and $h \in \mathbb{C}$. We claim that in such network, a destination connected to both relays should be able to decode all messages (note that with our choice of channel gains, there is no message for destinations that are only connected to relay ${\sf A}_2$).

Destinations that are connected to both relays receive the exact same signal (in the linear deterministic model), and the same codewords plus different noise terms (in the Gaussian model). Therefore, since each one of them is able to decode its message, then it should be able to decode the rest of the messages intended for destinations that are connected to both relays. They decode and remove such messages from the received signal. The remaining signal is the same codeword (plus different noise term in Gaussian model) received at the destinations that are only connected to relay ${\sf A}_1$. Therefore, those messages are also decodable at a destination that is connected to both relays.

We assume $1$-local view at the sources, therefore to achieve a normalized sum-rate of $\alpha$, each source should transmit at a rate greater than or equal to $\alpha n - \tau$ (in linear deterministic model). This is due to the fact that from each source's point of view, it is possible that the other S-D pairs have capacity $0$, therefore in order to achieve a normalized sum-rate of $\alpha$, it should transmit at a rate of at least $\alpha n - \tau$. The above argument alongside the MAC capacity at a destination connected to both relays, results in
\begin{equation}
d_{\mathrm{out}}({\sf A}_1) ( \alpha n - \tau ) \leq n \Rightarrow ( d_{\mathrm{out}}({\sf A}_1) \alpha - 1 ) n \leq d_{\mathrm{out}}({\sf A}_1) \tau.
\end{equation}
Since this has to hold for all values of $n$, and $\alpha$ and $\tau$ are independent of $n$, we get $\alpha \leq \frac{1}{d_{\mathrm{out}}({ \sf A}_1)}$.

In the Gaussian case, each source should transmit at a rate greater than or equal to $\alpha \log(1+|h|^2) - \tau$, since from each source's point of view, it is possible that the other S-D pairs have capacity $0$. Similar to the linear deterministic case, we get
\begin{equation}
d_{\mathrm{out}}({\sf A}_1) ( \alpha \log(1+|h|^2) - \tau ) \leq \log(1+d_{\mathrm{out}}({\sf A}_1) \times |h|^2),
\end{equation}
or equivalently
\begin{equation}
( d_{\mathrm{out}}({\sf A}_1) \alpha - 1 ) \log(1+|h|^2) \leq \log(d_{\mathrm{out}}({\sf A}_1)) + d_{\mathrm{out}}({\sf A}_1) \tau.
\end{equation}
Since this has to hold for all values of $h$, and $\alpha$ and $\tau$ are independent of $h$, we get $\alpha \leq \frac{1}{d_{\mathrm{out}}({\sf A}_1)}$.

Now that we have proved the converse for cases (1) and (2), we get
\begin{equation}
\alpha \leq \frac{1}{d_{\max}}.
\end{equation}
This completes the proof of the converse.

\noindent {\bf Achievability:} From Theorem~\ref{THM:CL}, we know that if we create a coded layer coloring with $T = d_{\max}$ in this netwrok, we can achieve the upper bound of $\frac{1}{d_{\max}}$. To do so, consider a set of $T = d_{\max}$ distinct colors, \emph{i.e.} $\mathcal{C} = \{ c_0,c_1,\ldots,c_{T-1} \}$.

Without loss of generality assume that $d_{\mathrm{in}}({\sf A}_1) \geq d_{\mathrm{in}}({\sf A}_2)$. Consider the route-expanded graph $\mathcal{G}_{\sf exp}$. To any node ${\sf V}_{i,j} \in \mathcal{V}_{\sf exp}$, we assign $\mathcal{C}_{i,j} = \emptyset$. We pick one member of $\mathcal{J}_1$ and one member of $\mathcal{J}_2$ randomly, and we assign to the corresponding sources the same color from $\mathcal{C}$ as their transmit color sets; we remove these members from $\mathcal{J}_1$ and $\mathcal{J}_2$. We keep picking two members and assign an unused member of $\mathcal{C}$ to the corresponding sources till $\mathcal{J}_2$ is empty. We assign to each remaining source an unused member of $\mathcal{C}$ randomly as its transmit color set. Note that to do so, we need $d_{\mathrm{in}}({\sf A}_1)$ number of colors. With this choice of color assignment, we have $|\mathcal{T}_{{\sf S}_i}| = 1$, $i=1,\ldots,K$.

In the second layer, we divide S-D pair ID's based on the connection of destinations to relays, \emph{i.e.} $\mathcal{J}^{\prime}_i$ is the set of all the S-D pair ID's such that the corresponding destination is connected to relay ${\sf A}_i$, $i=1,2$, and $\mathcal{J}^{\prime}_{12}$ is the set of all the other S-D pair ID's. Without loss of generality assume that $d_{\mathrm{out}}({\sf A}_1) \geq d_{\mathrm{out}}({\sf A}_2)$. In the second layer of the route-expanded graph $\mathcal{G}_{\sf exp}$, we pick one member of $\mathcal{J}^{\prime}_1$ and one member of $\mathcal{J}^{\prime}_2$ randomly, and we assign to the all nodes with these S-D pair ID's, the same color from $\mathcal{C}$ as their transmit color sets. We remove these members from $\mathcal{J}^{\prime}_1$ and $\mathcal{J}^{\prime}_2$. We keep picking two members and assign an unused member of $\mathcal{C}$ to the corresponding nodes in the route-expanded graph $\mathcal{G}_{\sf exp}$ till $\mathcal{J}^{\prime}_2$ is empty. We assign to the nodes corresponding to each remaining S-D pair ID an unused member of $\mathcal{C}$ randomly as their transmit color sets. Therefore, we need $d_{\mathrm{out}}({\sf A}_1)$ colors. At any node ${\sf V}_{i,j} \in \mathcal{V}_{\sf exp}$, we have $|\mathcal{T}_{{\sf V}_{i,j}}| = 1$. The total number of colors $T$ is therefore equal to the maximum degree of the nodes in $\mathcal{G}$, \emph{i.e.} $d_{\max}$. We also set
\begin{align}
\label{eq:receivesetK2K}
\mathcal{R}_{i,j} = \mathcal{T}_{i^\prime,j}: \forall {\sf V}_{i,j} \in \mathcal{V}_{\sf exp}, \text{ and } (i^\prime,i) \in \mathcal{E}.
\end{align}
Note that with this color assignment, we have 
\begin{align}
\label{eq:DreceivesetK2K}
\mathcal{R}_{{\sf D}_j} = \mathcal{T}_{i^\prime,j} \text{ such that } {\sf V}_{i^\prime,j} \text {is in the second layer}.
\end{align}

This coloring guarantees that the nodes with different pair ID's connected to the same receive node, are assigned different colors. Moreover, we only assign colors such that $\mathcal{C}_{i,j} = \emptyset$, $|\mathcal{T}_{i,j}|  = 1$, and $|\mathcal{R}_{{\sf D}_j}| = 1$. With the given argument, it is straight forward to verify that the described coloring satisfies conditions \textbf{C.1-C.7} in Section~\ref{section:CL}:
\begin{itemize}
\item \textbf{C.1} is satisfied since to any two nodes that have different S-D pair ID's, we have assigned different colors.

\item \textbf{C.2}, \textbf{C.3}, and \textbf{C.5} are satisfied since we have set $\mathcal{C}_{i,j} = \emptyset$ for all ${\sf V}_{i,j} \in \mathcal{V}_{\sf exp}$.

\item \textbf{C.4} is satisfied due to (\ref{eq:DreceivesetK2K}).

\item \textbf{C.6} is satisfied due to (\ref{eq:receivesetK2K}) and that all nodes in the same layer with the same S-D pair ID are assigned the same transmit color.

\item \textbf{C.7} is satisfied since with the given coloring at any node ${\sf V}_{i,j} \in \mathcal{V}_{\sf exp}$, there are no interferers.
\end{itemize}

Hence by Theorem~\ref{THM:CL}, we know that a normalized sum-rate of $\alpha = \frac{1}{T}$ is achievable. This completes the proof of the theorem for the case of $K \times 2 \times K$ networks. Note that are coloring is in fact, an Independent Layer coloring and hence, the normalized sum-capacity is achievable by MIL scheduling.

\noindent {\bf Extension}: Now, we need to extend our result to $K \times \underbrace{2 \times \ldots \times 2}_M \times K$ networks. Consider a non-interfering network as defined before. Then for the proof of converse, we set all the channel gains among relays equal to $n$ (in the linear deterministic model) and equal to $h$ (in the Gaussian model), and all the other channel gains as in the proof of the converse for the $K \times 2 \times K$ network. Since each destination is able to decode its message, each relay ${\sf V}_i \in \mathcal{V}$ can decode any message $\hbox{W}_j: j \in \mathcal{J}_{{\sf V}_i}$ (note that with the choice of channel gains, there exists at most one path for each S-D pair). Hence, we can consider all connected relays as one relay. 
Then the argument for the converse of the $K \times 2 \times K$ network is also valid here. An achievability similar to the one provided for the $K \times 2 \times K$ network works here.

\begin{figure}[ht]
\centering
\subfigure[]{\includegraphics[height = 3cm]{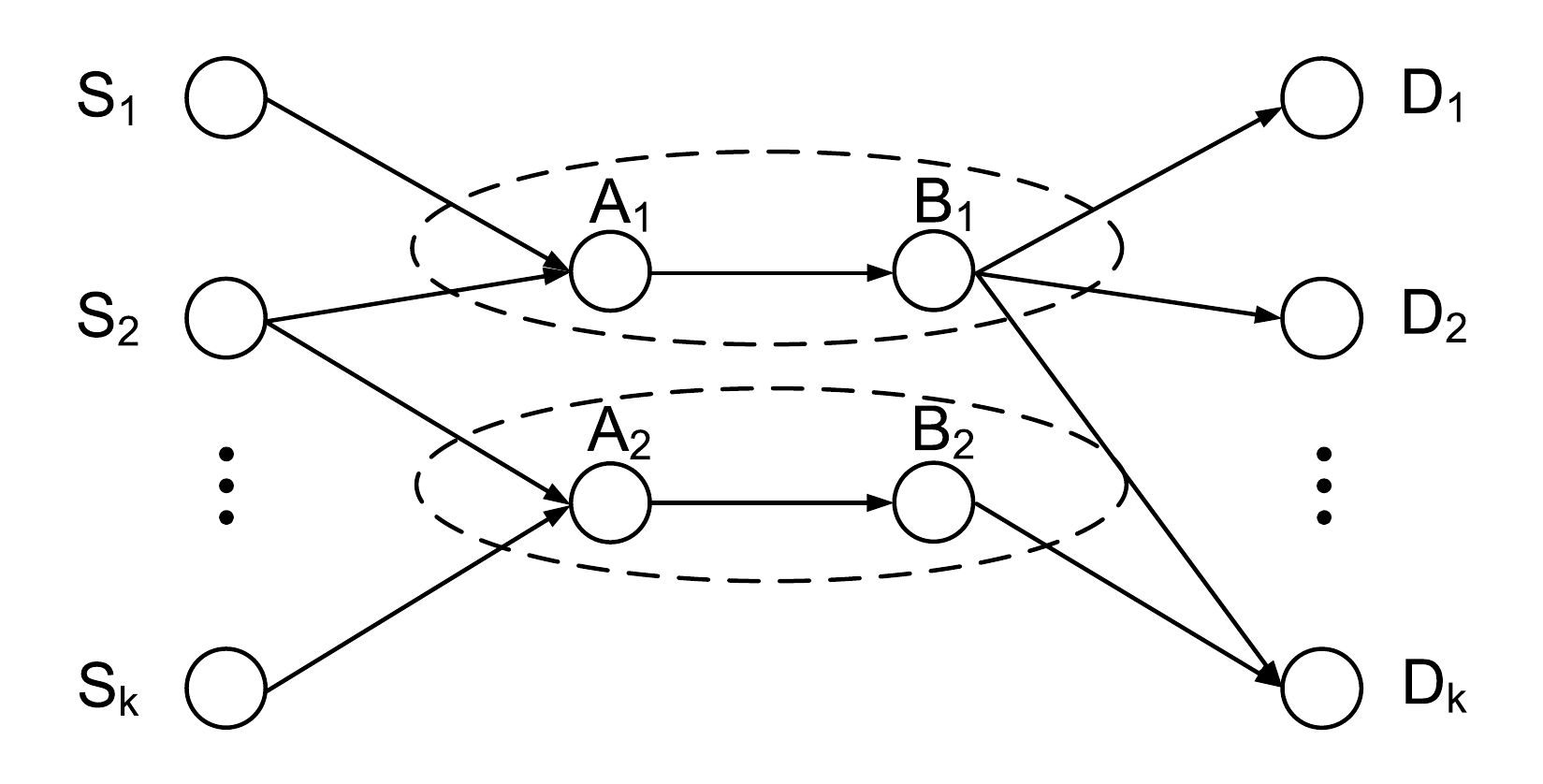}}
\hspace{0.5in}
\subfigure[]{\includegraphics[height = 3cm]{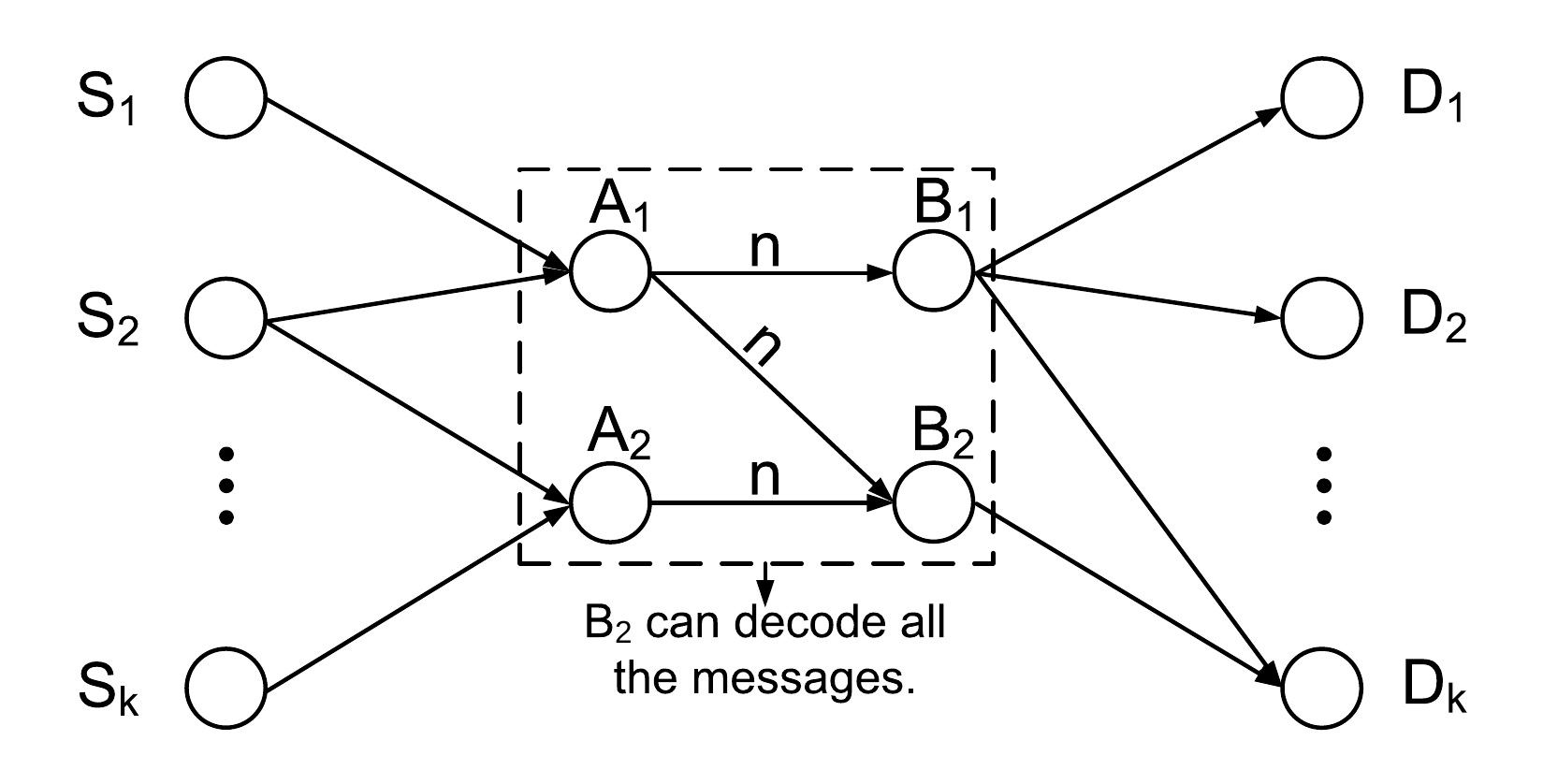}}
\caption{\it (a) A non-interfering $K \times 2 \times 2 \times K$ network, and (b) Relay ${\sf B}_2$ can decode all the messages.\label{k2kmultilayer}}
\end{figure}

If the condition does not hold, \emph{i.e.} the network is not non-interfering, we pick two paths from nodes in $\mathcal{V}_2$ to nodes in $\mathcal{V}_{L-1}$, one connecting 
${\sf V}^2_i$ to ${\sf V}^{M+1}_j$, and one connecting ${\sf V}^2_{\bar{i}}$ to ${\sf V}^{M+1}_{\bar{j}}$, $i \neq j$, $i,j \in \{ 1,2 \}$. Since the network is not non-interfering, there exists a path from ${\sf V}^2_i$ to ${\sf V}^{M+1}_{\bar{j}}$ or from ${\sf V}^2_{\bar{i}}$ to ${\sf V}^{M+1}_j$, pick such a path as well. 
Set all the channel gains on these three paths equal to $n$ (for the linear deterministic model) and equal to $h$ (for the Gaussian model). Use the channel gain assignment as in the proof of the converse for the $K \times 2 \times K$ network, for the links in the first and the last layer. Set all the other channel gains equal to $0$. 

With the assignment of channel gains described above, we can find a node ${\sf V}^l_{i^\ast}$ that is connected to all sources, see Figure~\ref{k2kmultilayer}(b). This node can decode all the messages of those S-D pairs $j$ such that $j \in \mathcal{J}_{{\sf V}^l_{i^\ast}}$ and $j \notin \mathcal{J}_{{\sf V}^l_{\bar{i^\ast}}}$. After decoding and removing these messages, node ${\sf V}^l_{i^\ast}$ receives the same codeword (with different noise term for the Gaussian model) as node ${\sf V}^l_{\bar{i^\ast}}$. Therefore, it should be able to decode all other messages as well. From the MAC upper bound at node ${\sf V}^l_{i^\ast}$ that can decode all messages, we have $\alpha \leq \frac{1}{K}$ (linear deterministic or Gaussian). We can achieve this upper bound by TDMA (note that TDMA is a special case of our coloring), and this completes the proof.

\end{proof}

\subsection{Folded-chain networks}

The previous theorem proved the optimality of interference avoidance techniques for a specific class of networks. In this subsection, we consider networks in which we need to incorporate coding in order to achieve the normalized sum-capacity with $1$-local view. We start with a single layer network that is motivated by a downlink cellular system as follows. Consider $3$ base stations (sources), \emph{i.e.} ${\sf S}_1, {\sf S}_2$ and ${\sf S}_3$, and $3$ receivers (destinatinations), \emph{i.e.} ${\sf D}_1, {\sf D}_2$ and ${\sf D}_3$, as in Figure~\ref{fig:downlink}. Each source ${\sf S}_i$ wishes to comunicate its message to destination ${\sf D}_i$, $i=1,2,3$. The coverage area for each one of the base stations is denoted by circle around it. Here, each destination is close to the boundaries of the coverage area of its correponding source, such that it only receives interference from one other source. Note that in this case, interference is comparable to the desired signal at each destination and hence, interference management is of crutial importance. The corresponding network $\mathcal{G}$ for the downlink cellular system described above, is depicted in Figure~\ref{fig:downlinkG}. From Section~\ref{section:examples}, we know that the normalized sum-capacity of this network with $1$-local view is $\alpha^\ast = \frac{1}{2}$ and can be achieved by implementing coding at the nodes; we also remind that interference avoidance techniques can only achieve a normalized sum-rate of $\alpha = \frac{1}{3}$ in this network. 

We extend the idea of such networks to $K$-user case via the following definition and as we will show, the normalized sum-capacity for this class of networks is achievable using repetition coding at sources.

\begin{figure}[ht]
\centering
\includegraphics[width = 5.5cm]{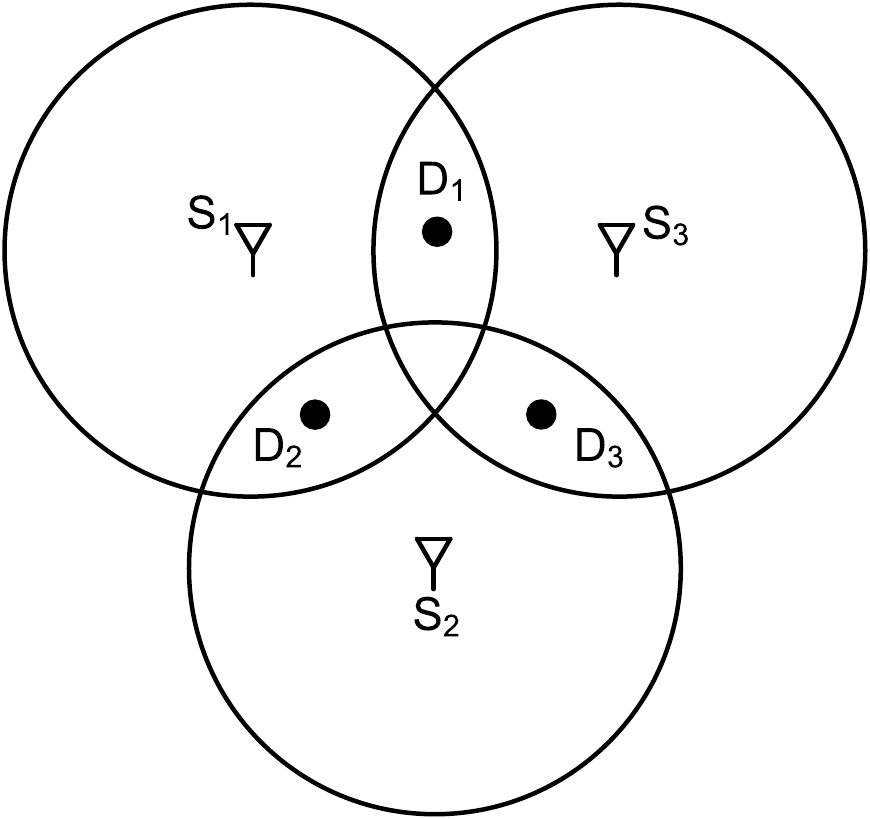}
\caption{\it Downlink cellular network with $3$ base stations and $3$ destinations. \label{fig:downlink}}
\end{figure}

\begin{figure}[ht]
\centering
\includegraphics[height = 3cm]{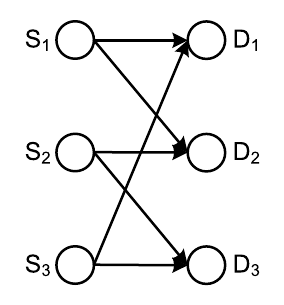}
\caption{\it The corresponding network $\mathcal{G}$ for the downlink cellular system illustrated in Figure~\ref{fig:downlink}.\label{fig:downlinkG}}
\end{figure}

\begin{definition}
A {\it single-layer $(K,m)$ folded-chain} network is a single-layer network with $K$ S-D pairs. In this network source ${\sf S}_i$ is connected to destinations with ID's $1 + \left[ \{ (i-1)^+ + (j-1) \} \textbf{ mod } K \right]$ where $i = 1,\ldots,K$, $j=1,\ldots,m$ and $(i-1)^+ = \max\{(i-1),0\}$. We assume $1 \le m \leq K$.
\end{definition}

\begin{lemma}
\label{THM:cyclic1}
The normalized sum-capacity of a single-layer $(K,m)$ folded-chain network (linear deterministic or Gaussian) with $1$-local view is $\alpha^\ast = \frac{1}{m}$ 
and is achieved by CL scheduling.
\end{lemma}

\begin{proof}

\noindent {\bf Converse}: Assume that a normalized sum-rate of $\alpha$ is achievable, \emph{i.e.} there exists a transmission strategy with $1$-local view, such that for all channel realizations, it achieves a sum-rate satisfying $\sum_{i=1}^K{R_i} \ge \alpha C_{\mathrm{sum}} -\tau$ with error probabilities going to zero as $N \rightarrow \infty$ and for some constant $\tau \in R$ independent of the channel gains. We will show that $\alpha \leq \frac{1}{m}$.


The proof of the converse for $K=1$ is trivial. Consider a single-layer $(K,m)$ folded-chain network, where the channel gain of a link from source $i$ to destinations $i, i+1, \ldots, m$ is equal to $n$ (for the linear deterministic model) or $h$ (for the Gaussian model), $i=1,2,\ldots,m$, and all the other channel gains are equal to zero, where $n \in \mathbb{N}$ and $h \in \mathbb{C}$. See Figure~\ref{fig:kdcyclic} for a depiction.

\begin{figure}[ht]
\centering
\includegraphics[width = 3.5cm]{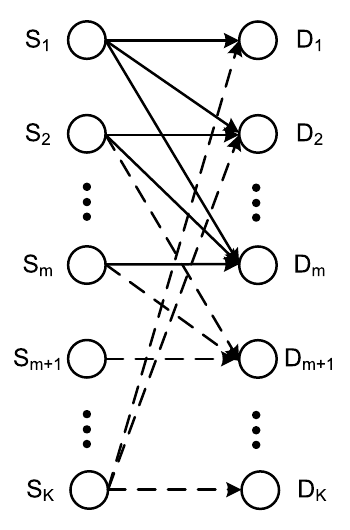}
\caption{\it Channel gain assignment in a single-layer $(K,m)$ folded-chain network. All solid links have capacity $n$ (for the linear deterministic model) or $h$ (for the Gaussian model), and all dashed links have capacity $0$. \label{fig:kdcyclic}}
\end{figure}

Suppose, a normalized sum-rate of $\alpha$ is feasible in this network with $1$-local view. Each source due to its local view of the network, should transmit at a rate greater than or equal to $\alpha n - \tau$ (for the linear deterministic model) or $\alpha \log \left( 1 + |h|^2 \right) - \tau$ (for the Gaussian model), to guarantee a normalized sum-rate of $\alpha$. Destination $1$ receives no interference and decodes its message. Destination $2$, decodes its message and removes it from the received signal, what is left is exactly the same as what destination $1$ receives, therefore destination $2$ is able to decode $\hbox{W}_1$ and $\hbox{W}_2$. If we continue this argument, we see that destination ${\sf D}_m$ is be able to decode all $\hbox{W}_i$'s, $i=1,2,\ldots,m$. The MAC capacity at destination ${\sf D}_m$, for the linear deterministic model gives us
\begin{equation}
m ( \alpha n - \tau ) \leq n \Rightarrow ( m \alpha - 1 ) n \leq d \tau.
\end{equation}
Since this has to hold for all values of $n$, and $\alpha$ and $\tau$ are independent of $n$, we get $\alpha \leq \frac{1}{m}$. For the Gaussian model, the MAC capacity at destination ${\sf D}_m$ gives us
\begin{equation}
m ( \alpha \log(1+|h|^2) - \tau ) \leq \log(1+ m \times |h|^2),
\end{equation}
which results in
\begin{equation}
( m \alpha - 1 ) \log(1+|h|^2) \leq \log(m) + m \tau.
\end{equation}
Since this has to hold for all values of $h$, and $\alpha$ and $\tau$ are independent of $h$, we get $\alpha \leq \frac{1}{m}$.

\noindent {\bf Achievability}: We will present a coded layer coloring with $m$ colors for the single-layer $(K,m)$ folded-chain network. Then by Theorem~\ref{THM:CL}, we know that the upper bound of $\frac{1}{m}$ is achievable. Suppose $\mathcal{C} = \{ c_0, c_1, \ldots, c_{m-1} \}$, and assume that $m < K < 2m$ (we will later generalize the achievability scheme for arbitrary $K$). Note that the route-expanded graph of this network is the same as itself (since we have a single-layer network). Let $m^\prime = K - m + 1$. To each source ${\sf S}_i$, $i=1,\ldots,K$, we assign $\mathcal{C}_{{\sf S}_i} = \emptyset$ and $\mathcal{T}_{{\sf S}_i}$ as follows,
\begin{equation}
\label{eq:Kmcyclictransmitset}
c_j \in \mathcal{T}_{{\sf S}_i} \Leftrightarrow j+1 \leq i \leq j+m^{'}, \quad j = 0,1,\ldots,m-1.
\end{equation}

We also set
\begin{equation}
\label{eq:Kmcyclicreceiveset}
\mathcal{R}_{{\sf D}_i} =
\left\{ \begin{array}{ll}
\{ c_{i-1} \} & \text{if }1 \le i \le m \\
\{ c_0, \ldots, c_r \} & \text{if } i = m + r, 1 \le r \le K-m \\
\end{array} \right.
\end{equation}
and $\mathcal{R}_{{\sf D}_i} = \emptyset$, $i=1,\ldots,K$.

Since $\mathcal{C}_{{\sf S}_i} = \emptyset$, for any $i \in \{1,\ldots,K\}$, it is straight forward to verify that this coloring satisfies conditions \textbf{C.1-C.7} in Section~\ref{section:CL}:
\begin{itemize}
\item \textbf{C.1} is satisfied since to any two nodes that have different S-D pair ID's, we have assigned different colors.

\item \textbf{C.2}, \textbf{C.3}, and \textbf{C.5} are satisfied since we have set $\mathcal{C}_{i,j} = \emptyset$ for all ${\sf V}_{i,j} \in \mathcal{V}_{\sf exp}$.

\item \textbf{C.4} is satisfied due to (\ref{eq:Kmcyclictransmitset}) and (\ref{eq:Kmcyclicreceiveset}).

\item \textbf{C.6} is satisfied since from (\ref{eq:Kmcyclictransmitset}) and (\ref{eq:Kmcyclicreceiveset}) we have $|\mathcal{R}_{{\sf D}_i} \cap \mathcal{T}_{{\sf S}_i}| = 1$, $i=1,\ldots,K$.

\item \textbf{C.7} is satisfied since with the given coloring at any node ${\sf V}_{i,j} \in \mathcal{V}_{\sf exp}$ due to (\ref{eq:Kmcyclictransmitset}), all interferers share  a common color in their transmit color sets, which is in $\mathcal{R}_{{\sf D}_i} \setminus \mathcal{T}_{{\sf S}_i}$.
\end{itemize}

Therefore from Theorem~\ref{THM:CL}, we know that we can achieve $\alpha = \frac{1}{m}$. 

For general $K$ the achievability works as follows. Suppose, $K = c (2m - 1) + r$, where $c \geq 1$ and $0 \leq r < (2 m - 1)$, we implement the scheme for S-D pairs $1,2,\ldots,2m-1$ as if they are the only pairs in the network. The same for source-destination pairs $2m, 2m + 1,\ldots, 4m -2$ and etc. Finally, for the last $r$ S-D pairs, we implement the scheme with $m^{'} = \max \{ r - m + 1, 1 \}$. This completes the proof of the lemma.
\end{proof}

\begin{remark}
MIL coloring for single-layer networks is the same as a vertex coloring for the corresponding route-adjacency graph such that any two connected vertices are assigned different colors. Therefore, for a single-layer $(K,m)$ folded-chain network with $1$-local view, MIL scheduling achieves a normalized sum-rate of $\alpha = \frac{1}{\xi}$ where $\xi$ is the chromatic number of the correponding route-adjacency graph. For instance, for a single-layer $(K,m)$ folded-chain network with $1$-local view where $\frac{K}{2} \leq m \leq K$, MIL scheduling achieves a normalized sum-rate of $\alpha = \frac{1}{K}$, whereas MCL scheduling achieves a normalized sum-rate of $\alpha = \frac{1}{m} \geq \frac{1}{K}$ (note that for $\frac{K}{2} \leq m \leq K$, the route-adjacency graph of a single-layer $(K,m)$ folded-chain network is a complete graph). As another example, for a single-layer $(K,2)$ folded-chain network with $1$-local view where $K$ is an even number, MIL scheduling achieves a normalized sum-rate of $\alpha = \frac{1}{2}$ which turns out to be the normalized sum-capacity (note that the route-adjacency graph of a single-layer $(K,2)$ folded-chain network is a cycle and when $K=2$, the chromatic number is $\xi = 2$).
\end{remark}

\begin{remark}
Although MIL scheduling is not necessarily optimal for the single-layer $(K,m)$ folded-chain networks, in some cases, slight modifications of topology can result in a network where MIL scheduling is optimal. For instance consider a single-layer $(K,2)$ folded-chain network where $K$ is an odd number, see Figure~\ref{fig:K2odd}(a). If we remove the edge from source ${\sf S}_K$ to destination ${\sf D}_1$, we will have a network for which chromatic number is $\xi = 2$, see Figure~\ref{fig:K2oddremove}(a) and Figure~\ref{fig:K2oddremove}(b). MIL scheduling achieves a normalized sum-rate of $\alpha= \frac{1}{2}$ for the network in Figure~\ref{fig:K2oddremove}(a) with $1$-local view, and from Lemma~\ref{lemma:generalhalf}, we know that this is in fact the normalized sum-capacity in this case.
\end{remark}

\begin{figure}[ht]
\centering
\subfigure[]{\includegraphics[height = 4cm]{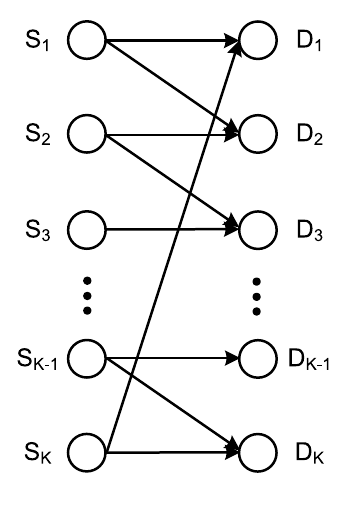}}
\hspace{0.4in}
\subfigure[]{\includegraphics[height = 3.5cm]{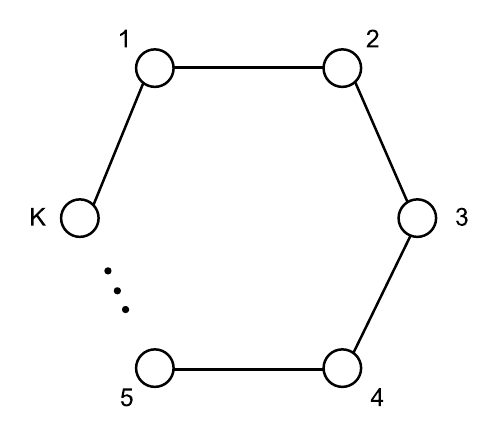}}
\caption{\it (a) A single-layer $(K,2)$ folded-chain network with $K$ odd, and (b) its conflict graph with $\xi = 3$.\label{fig:K2odd}}
\end{figure}

\begin{figure}[ht]
\centering
\subfigure[]{\includegraphics[height = 4cm]{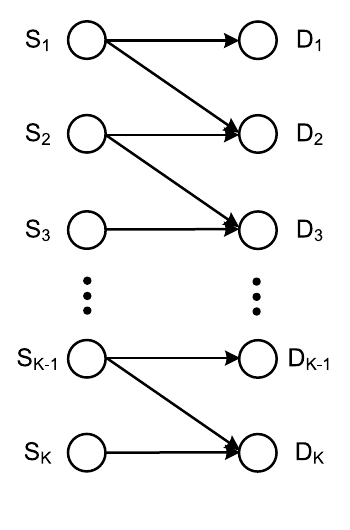}}
\hspace{0.4in}
\subfigure[]{\includegraphics[height = 3.5cm]{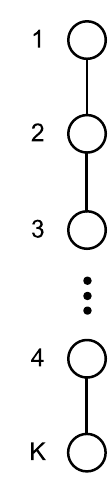}}
\caption{\it (a) The network created by removing an edge from a single-layer $(K,2)$ folded-chain network with $K$ odd, and (b) its conflict graph with $\xi = 2$.\label{fig:K2oddremove}}
\end{figure}

In the previous lemma, we investigated a class of networks in which repetition coding achieves the normalized sum-capacity. Next, we consider a class of networks for which we need to incorporate network coding in order to achieve the normalized sum-capacity with $1$-local view. This class of networks is created by cascading two single-layer $(K,m)$ folded-chain networks as defined below.

\begin{definition}
A {\it two-layer $(K,m)$ folded-chain} network is a two-layer network with $K$ S-D pairs and $K$ relays in the middle. Each S-D pair $i$ has $m$ disjoint paths, through relays with index $1 + \left[ \{ (i-1)^+ + (j-1) \} \textbf{ mod } K \right]$ where $i = 1,\ldots,K$, $j=1,\ldots,m$. We assume $1 \le m \leq K$.
\end{definition}

\begin{theorem}
\label{THM:cyclic2}
The normalized sum-capacity of a two-layer $(K,m)$ folded-chain network (linear deterministic or Gaussian) with $1$-local view is upper bounded by $\alpha = \frac{1}{m}$ 
and is achievable by CL scheduling for $m \in \{ 1,2,K-1,K \}$. 
\end{theorem}

\begin{proof}

\noindent {\bf Converse}: For the two-layer $(K,m)$ folded-chain network, in the first layer use the same channel gain assignment described for proof of Lemma~\ref{THM:cyclic1}, and in the second layer, we set the channel gain from relay $i$ to destination $i$ equal to $n$ (for the linear deterministic model) or $h$ (for the Gaussian model), and all other channel gains equal to $0$, $i = 1,2,\ldots,m$, see Figure~\ref{fig:kdcyclic2}. With this configuration, each destination ${\sf D}_i$ is {\it only} connected to relay ${\sf A}_i$, $i=1,\ldots,m$, and each relay ${\sf A}_i$ has all the information that destination ${\sf D}_i$ requires in order to decode its message, $i=1,\ldots,m$. Therefore, the converse argument presented for the single-layer case is valid here, and hence we have $\alpha \leq \frac{1}{m}$.

\begin{figure}[ht]
\centering
\includegraphics[width = 4.5cm]{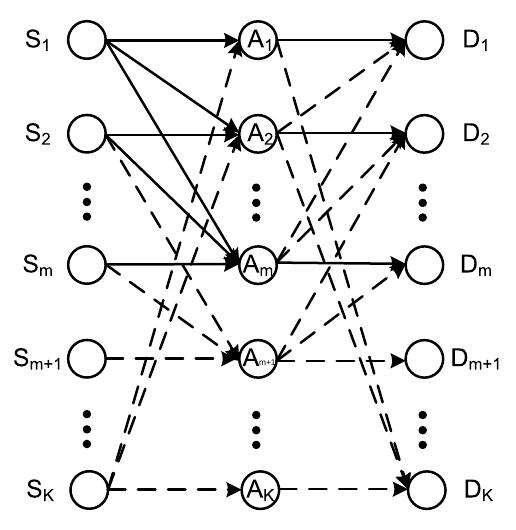}
\caption{\it Channel gain assignment in a two-layer $(K,m)$ folded-chain network. All solid links have capacity $n$ (for the linear deterministic model) or $h$ (for the Gaussian model), and all dashed links have capacity $0$. \label{fig:kdcyclic2}}
\end{figure}

\noindent {\bf Achievability}: The result is trivial for $m=1$. For $m=K$, the upper bound of $\alpha = \frac{1}{K}$ can be achieved by TDMA.

\noindent $\bullet$ $m=2$: Suppose $\mathcal{C} = \{ c_0, c_1 \}$. We have two cases: (a) $K$ is even: in this case, to any node ${\sf V}_{i,j} \in \mathcal{V}_{\sf exp}$ assign $\mathcal{C}_{i,j} = \emptyset$ and $\mathcal{T}_{i,j} = \mathcal{R}_{i,j} = \{ c_0 \}$ if $j$ is odd and $\mathcal{T}_{i,j} = \mathcal{R}_{i,j} = \{ c_1 \}$ if $j$ is even. The achievability scheme in this case turns out to be end-to-end interference avoidance; (b) $K$ is odd: in this case for $j \in \{ 1,\ldots,K-1 \}$, assign to any node ${\sf V}_{i,j} \in \mathcal{V}_{\sf exp}$, $\mathcal{C}_{i,j} = \emptyset$ and $\mathcal{T}_{i,j} = \mathcal{R}_{i,j} = \{ c_0 \}$ if $j$ is odd and $\mathcal{T}_{i,j} = \mathcal{R}_{i,j} = \{ c_1 \}$ if $j$ is even. For $j=K$, to any node ${\sf V}_{i,K} \in \mathcal{V}_{\sf exp}$, assign $\mathcal{C}_{i,K} = \mathcal{T}_{i,j^\prime}$ and $\mathcal{T}_{i,j} = \mathcal{C} \setminus \mathcal{T}_{i,j^\prime}$ and $\mathcal{R}_{i,K} = \{ c_0, c_1 \}$ where $j^\prime \in \mathcal{J}_{{\sf V}_i}$ and $j^\prime \neq K$. With this coloring, S-D pairs $1,\ldots,K-1$, communicate interference-free and we use network coding for S-D pair $K$. Note that we only have two nodes in the route-expanded graph where $\mathcal{C}_{i,j} \neq \emptyset$, and there are no interferers at any node ${\sf V}_{i,j} \in \mathcal{V}_{\sf exp}$. It is straightforward to verify that conditions \textbf{C.1-C.7} are satisfied and hence, by Theorem~\ref{THM:CL} a normalized sum-rate of $\alpha = \frac{1}{2}$ is achievable.

\noindent $\bullet$ $m=K-1$: Suppose $\mathcal{C} = \{ c_0, \ldots, C_{K-2} \}$. In this case for $j \in \{ 1,\ldots,K-1 \}$, assign to any node ${\sf V}_{i,j} \in \mathcal{V}_{\sf exp}$, $\mathcal{C}_{i,j} = \emptyset$ and $\mathcal{T}_{i,j} = \mathcal{R}_{i,j} = \{ c_{j-1} \}$, and to any node ${\sf V}_{i,K} \in \mathcal{V}_{\sf exp}$, assign $\mathcal{C}_{i,K} = \bigcup_{\small j^\prime \in \mathcal{J}_{{\sf V}_i}, j^\prime \neq K}{\mathcal{T}_{i,j^\prime}}$ and $\mathcal{T}_{i,j} = \mathcal{C} \setminus \mathcal{C}_{i,K}$ and $\mathcal{R}_{i,K} = \mathcal{C}$. Note that we only have two nodes in the route-expanded graph where $\mathcal{C}_{i,j} \neq \emptyset$, and there are no interferers at any node ${\sf V}_{i,j} \in \mathcal{V}_{\sf exp}$. It is straightforward to verify that conditions \textbf{C.1-C.7} are satisfied and hence, by Theorem~\ref{THM:CL} a normalized sum-rate of $\alpha = \frac{1}{m}$ is achievable.
\end{proof}

\subsection{MCL scheduling vs. MIL scheduling: nested folded-chain networks}

As we already know, MIL scheduling is a special subclass of MCL scheduling. Moreover in this subsection, we show that the gain from using MCL scheduling over MIL scheduling can be unbounded. To do so, we first define the following class of networks.
\begin{definition}
An {\it $L$-nested folded-chain} network is a single-layer network with $K = 3^L$ S-D pairs, $\{ {\sf S}_1, \ldots, {\sf S}_{3^L} \}$ and $\{ {\sf D}_1, \ldots, {\sf D}_{3^L} \}$. For $L=1$, an $L$-nested folded-chain network is the same as a single-layer $(3,2)$ folded-chain network. For $L > 1$, an $L$-nested folded-chain network is formed by first creating $3$ copies of an $(L-1)$-nested folded-chain network. Then,
\begin{itemize}
\item The $i$-th source in the first copy is connected to the $i$-th destination in the second copy, $i=1,\ldots,3^{L-1}$,
\item The $i$-th source in the second copy is connected to the $i$-th destination in the third copy, $i=1,\ldots,3^{L-1}$,
\item The $i$-th source in the third copy is connected to the $i$-th destination in the first copy, $i=1,\ldots,3^{L-1}$.
\end{itemize}
Figure~\ref{fig:nested} illustrates a $2$-nested folded-chain network.
\end{definition}

\begin{figure}[ht]
\centering
\includegraphics[height = 7cm]{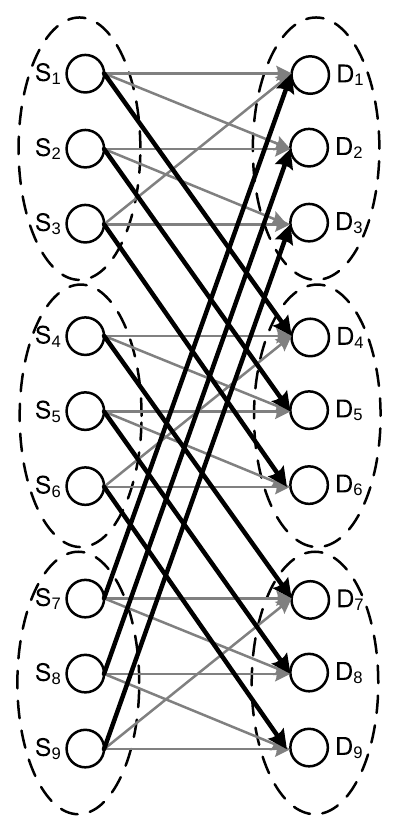}
\caption{\it A $2$-nested folded-chain network. \label{fig:nested}}
\end{figure}

Consider an $L$-nested folded-chain network. The conflict graph of this network is fully connected, and as a result, MIL scheduling achieves a normalized sum-rate of $\left( \frac{1}{3} \right)^L$. However, we know that for a single-layer $(3,2)$ folded-chain network, MCL scheduling results in a normalized sum-rate of $\frac{1}{2}$. Hence, 
applying CL scheduling to an $L$-nested folded-chain network, a normalized sum-rate of $\left( \frac{1}{2} \right)^L$ is achievable. For instance, consider the $2$-nested folded-chain network in Figure~\ref{fig:nested}. Consider the induced subgraphs of all three S-D pairs. We show that any transmission strategy over these induced subgraphs can be implemented in the original network by using only four time-slots, such that all nodes receive the same signal as if they were in the induced subgraphs. To achieve a normalized sum-rate of $\alpha = \left( \frac{1}{2} \right)^2$, we split the communication block into $4$ time-slots of equal length. During time-slot $1$, sources 1, 2, 4, and 5 transmit the same codewords as if they are in the induced subgraphs. During time-slot $2$, sources 3 and 6 transmit the same codewords as if they are in the induced subgraphs and sources 2 and 5 repeat their transmit signal from the first time-slot. During time-slot $3$, sources 7 and 8 transmit the same codewords as if they are in the induced subgraphs and sources 4 and 5 repeat their transmit signal from the first time-slot. During time-slot $4$, source 9 transmits the same codewords as if it is in the induced subgraph and sources 5, 6, and 8 repeat their transmit signal. It is straight forward to verify that with this scheme, all destinations receive the same signal as if they were in the induced subgraphs. Hence, a normalized sum-rate of $\alpha = \left( \frac{1}{2} \right)^2$ is achievable for the network in Figure~\ref{fig:nested}. Therefore, the gain of using MCL scheduling over MIL scheduling is $\left( \frac{3}{2} \right)^L$ which goes to infinity as $L \rightarrow \infty$. As a result, we can state the following lemma.

\begin{lemma}
Consider an $L$-nested folded-chain network. The gain of using MCL scheduling over MIL scheduling is $\left( \frac{3}{2} \right)^L$ which goes to infinity as $L \rightarrow \infty$.
\end{lemma}

\section{Concluding Remarks}
\label{conclusion}
In this paper, we developed a new transmission strategy (Coded Layer scheduling) for multi-layer wireless networks with partial netwrok knowledge (i.e., $1$-local view) that combines multiple ideas including interference avoidance and network coding. We established the optimality of our proposed strategy and its special case (Independent Layer scheduling) for some classes of networks, in terms of achieving the normalized sum-capacity. We also demonstrated several connections between network topology, normalized sum-capacity, and the achievability strategies. However, it remains open to evaluate the performance of Coded Layer scheduling in a more general class of networks.

So far, we have only studied cases with $1$-local view. One major direction is to characterize the increase in normalized sum-capacity as nodes learn more about the network (i.e., $h$-local view, $h>1$). We have also focused on the case that the nodes know the network-connectivity globally, but the actual values of the channel gains are only known for a subset of flows. Another important direction would be to understand the effects of local knowledge about network connectivity on the capacity and develop distributed strategies to optimally route information with partial knowledge about network connectivity. Finally, it is quite interesting to see how CL scheduling can be improved by incorporating more general network codes.

\section{Acknowledgement}
The authors would like to thank Kanes Sutuntivorakoon of Rice University for helpful discussions related to Section V.C of this paper.
\appendices

\section{Proof of Lemma~\ref{lemma:generalhalf}}
\label{Appendix:generalhalf}
Consider a path from source ${\sf S}_i$ to destination ${\sf D}_j$, $i \neq j$, $i,j \in \{1,\ldots,K\}$. Assign channel gain of $n$ (for the linear deterministic model) and $h$ (for the Gaussian model) to all edges in this path. For each one of the two S-D pairs $i$ and $j$, pick exactly one path from the source to the destination and assign channel gain of $n$ (for the linear deterministic model) or $h$ (for the Gaussian model) to all edges in these paths. Assign channel gain of $0$ to all remaining edges in the network $\mathcal{G}$, where $n \in \mathbb{N}$ and $h \in \mathbb{C}$. See Figure~\ref{fig:upperbound} for an illustration. 

\begin{figure}[h]
\centering
\includegraphics[height = 3.5cm]{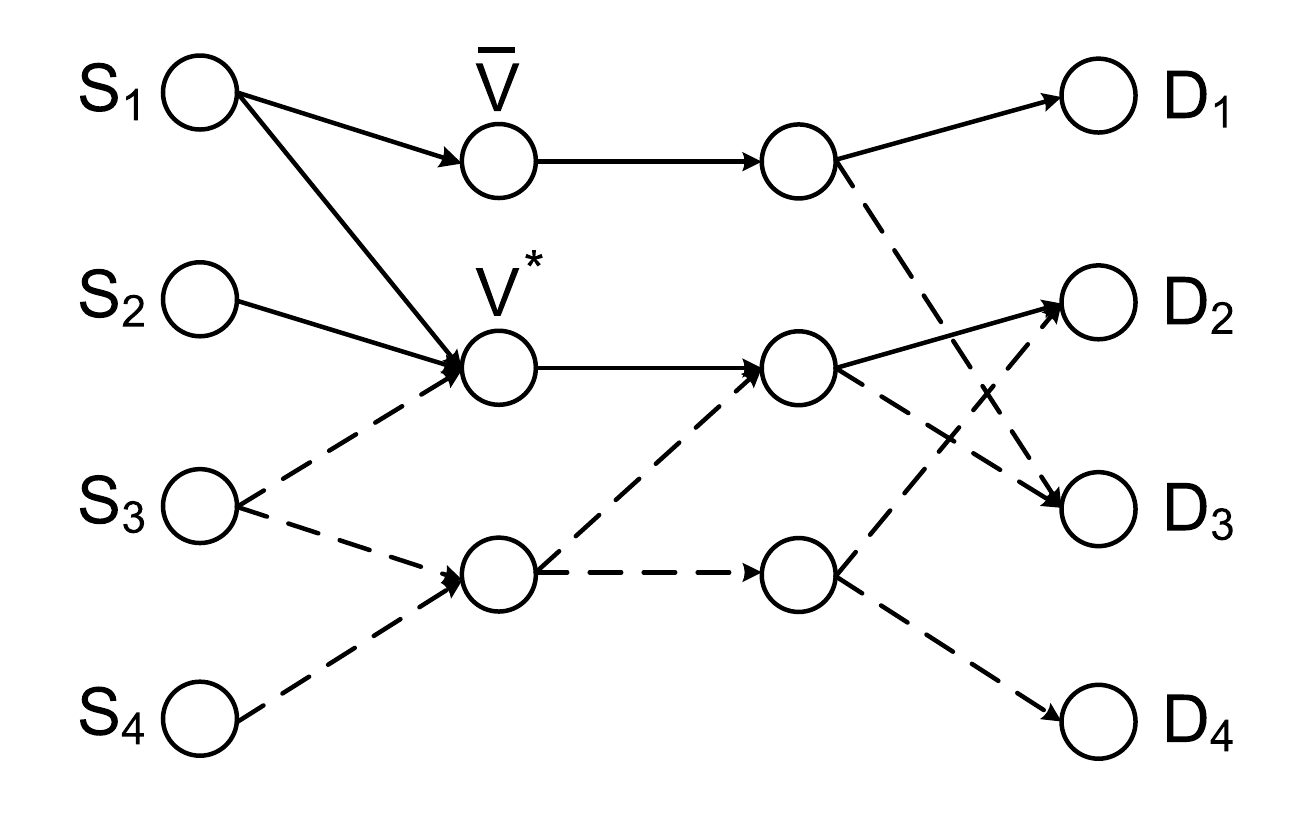}
\caption{\it A path exists from source $\sf S_1$ to destination $\sf D_2$; all solid edges have capacity $n$ (for linear deterministic model) or $h$ (for the Gaussian model) and the rest have capacity $0$. \label{fig:upperbound}}
\end{figure}


In order to guarantee a normalized sum-rate of $\alpha$ with $1$-local view, each source has to transmit at a rate greater than or equal to $\alpha n - \tau$ (for the linear deterministic model) or $\alpha \log(1+|h|^2) - \tau$ (for the Gaussian model). This is due to the fact that from each source's point of view, it is possible that the other S-D pairs have capacity $0$. Suppose this rate is feasible, \emph{i.e.} destinations ${\sf D}_i$ and ${\sf D}_j$ are able to decode $\hbox{W}_i$ and $\hbox{W}_j$ respectively.

Since there exists a path from source ${\sf S}_i$ to destination ${\sf D}_j$, $i \neq j$, $i,j \in \{1,\ldots,K\}$, we can find a node ${\sf V}^\ast \in \mathcal{V}$ such that ${\sf V}^\ast \in \mathcal{G}_{ij}$ and ${\sf V}^\ast \in \mathcal{G}_{jj}$, see Figure~\ref{fig:upperbound}. Node ${\sf V}^\ast$ is able to decode $\hbox{W}_j$ since it has more information than destination ${\sf D}_j$. Node ${\sf V}^\ast$ decodes $\hbox{W}_j$ and removes it from the received signal. Now, for the Gaussian model, it has statistically the same received signal as node $\bar{\sf V}$ (a node in $\mathcal{G}_{ii}$ and in the same layer as ${\sf V}^\ast$) and for the linear deterministic model it has the same received signal as node $\bar{\sf V}$. Node $\bar{\sf V}$ is able to decode $\hbox{W}_i$ since it has more information than destination ${\sf D}_i$., as a result, ${\sf V}^\ast$ is also able to decode $\hbox{W}_i$. This means that there exists a node ${\sf V}^\ast \in \mathcal{V}$ that can decode both $\hbox{W}_i$ and $\hbox{W}_j$. 

Hence, the MAC capacity at ${\sf V}^\ast$ in the linear deterministic model gives us
\begin{equation}
2 \alpha n - 2 \tau \leq n \Rightarrow (2 \alpha - 1) n \leq 2 \tau.
\end{equation}
Since this has to hold for all values of $n$, and $\alpha$ and $\tau$ are independent of $n$, we get $\alpha \leq \frac{1}{2}$.

For the Gaussian model, the MAC capacity at ${\sf V}^\ast$ gives us
\begin{equation}
2 ( \alpha \log(1+|h|^2) - \tau ) \leq \log(1 + 2|h|^2),
\end{equation}
which results in
\begin{equation}
( 2 \alpha - 1 ) \log(1+|h|^2) \leq 1 + 2 \tau.
\end{equation}
Since this has to hold for all values of $h$, and $\alpha$ and $\tau$ are independent of $h$, we get $\alpha \leq \frac{1}{2}$.

\section{}
\label{Appendix:kappa}
\begin{claim}
\label{claim:kappa}
Consdier a complex multi-hop Gaussian relay network with one source ${\sf S}$ and one destination ${\sf D}$ represented by a directed graph $\mathcal{G} = (\mathcal{V}, \mathcal{E},\{ h_{ij} \}_{(i,j) \in \mathcal{E}})$ where $\{ h_{ij} \}_{(i,j) \in \mathcal{E}}$ represent the channel gains associated with the edges.

We assume that at each receive node the additive white complex Gaussian noise has variance $\sigma^2$. We also assume a power constraint of $P$ at all nodes, \emph{i.e.} $\lim_{n\to\infty}\frac{1}{n} \mathbb{E}(\sum_{t=1}^n{|X_{{\sf V}_{i}}[t]|^2}) \le P$. 
Denote the capacity of this network by $C(\sigma^2,P)$. Then, for all $T \geq 1$, $T \in \mathbb{R}$, we have 
\begin{equation}
C(\sigma^2,P) - \tau \le C(T\sigma^2,P/T) \le C(\sigma^2,P),
\end{equation}
where $\tau = |\mathcal{V}| \left( 2 \log T + 17 \right)$ is a constant independent of channel gains, $P$, and $\sigma^2$.
\end{claim}

\begin{proof}
First note that by increasing noise variances and decreasing power constraint, we only decrease the capacity, hence $C(T\sigma^2,P/T) \le C(\sigma^2,P)$. To prove the other inequality, we use the results in \cite{ADTJ10}. The cut-set bound $\bar{C}$ is defined as
\begin{align}
\bar{C}(\sigma^2,P) = \max_{p(\{ X_j \}_{{\sf V}_j \in \mathcal{V}})}{\min_{\Omega \in \Lambda_D}{I(Y_{\Omega^c};X_{\Omega}|X_{\Omega^c})}},
\end{align}
where $\Lambda_D = \{ \Omega : {\sf S} \in \Omega, {\sf D} \in \Omega^c \}$ is the set of all S-D cuts. \footnote{ A cut $\Omega$ is a subset of  $\mathcal{V}$ such that ${\sf S} \in \Omega, {\sf D} \notin \Omega$, and $\Omega^c = \mathcal{V} \setminus \Omega$.} Also $\bar{C}_{i.i.d}(\sigma^2,P) = \min_{\Omega \in \Lambda_D}{\log|{\bf I} + \frac{P}{\sigma^2} {\bf G}_{\Omega} {\bf G}^{\ast}_{\Omega}|}$ is the cut-set bound evaluated for i.i.d. $\mathcal{N}(0,P)$ input distributions and ${\bf G}_{\Omega}$ is the transfer matrix associated with the cut $\Omega$, \emph{i.e.} the matrix relating the vector of all the inputs at the nodes in $\Omega$, denoted by ${\bf X}_\Omega$, to the vector of all the outputs in $\Omega^c$, denoted by ${\bf Y}_{\Omega^c}$, as in ${\bf Y}_{\Omega^c} = {\bf G}_{\Omega} {\bf X}_\Omega + {\bf Z}_{\Omega^c}$ where ${\bf Z}_{\Omega^c}$ is the noise vector. In \cite{ADTJ10}, it has been shown that 
\begin{align}
\label{eq:appendix1}
\bar{C}_{i.i.d}(\sigma^2,P) - 15 |\mathcal{V}| \le C(\sigma^2,P) \le \bar{C}_{i.i.d}(\sigma^2,P) + 2 |\mathcal{V}|,
\end{align}
where $|\mathcal{V}|$ is the total number of nodes in the network. Similarly, we have
\begin{align}
\label{eq:appendix2}
& \bar{C}_{i.i.d}(T\sigma^2,P/T) - 15 |\mathcal{V}| \le C(T\sigma^2,P/T) \nonumber \\
& C(T\sigma^2,P/T) \le \bar{C}_{i.i.d}(T\sigma^2,P/T) + 2 |\mathcal{V}|.
\end{align}

Now, we will show that
\begin{align}
C(\sigma^2,P) - C(T\sigma^2,P/T) \le |\mathcal{V}| \left( 2 \log T + 17 \right).
\end{align}
For any S-D cut $\Omega \in \Lambda_D$, $\frac{P}{\sigma^2} {\bf G}_{\Omega} {\bf G}^{\ast}_{\Omega}$ is a positive semi-definite matrix. Hence, there exists a unitary matrix ${\bf U}$ such that ${\bf U} {\bf G}_{diag} {\bf U}^\ast = \frac{P}{\sigma^2} {\bf G}_{\Omega} {\bf G}^{\ast}_{\Omega}$ where ${\bf G}_{diag}$ is a diagonal matrix. Refer to the non-zero elements in ${\bf G}_{diag}$ as $g_{ii}$'s. We then have
\begin{align}
\label{eq:monotone}
& \log|{\bf I} + \frac{P}{\sigma^2} {\bf G}_{\Omega} {\bf G}^{\ast}_{\Omega}| - \log|{\bf I} + \frac{P}{T^2\sigma^2} {\bf G}_{\Omega} {\bf G}^{\ast}_{\Omega}| \nonumber \\
& = \log|{\bf I} + {\bf U} {\bf G}_{diag} {\bf U}^\ast| - \log|{\bf I} + \frac{1}{T^2} {\bf U} {\bf G}_{diag} {\bf U}^\ast| \nonumber \\
& = \log|{\bf U} {\bf U}^\ast + {\bf U} {\bf G}_{diag} {\bf U}^\ast| - \log|{\bf U} {\bf U}^\ast + \frac{1}{T^2} {\bf U} {\bf G}_{diag} {\bf U}^\ast| \nonumber \\
& = \log \left( |{\bf U}| |{\bf I}+ {\bf G}_{diag}| |{\bf U}^\ast| \right) - \log \left( |{\bf U}| |{\bf I}+ \frac{1}{T^2} {\bf G}_{diag}| |{\bf U}^\ast| \right) \nonumber \\
& = \log|{\bf I}+ {\bf G}_{diag}| - \log|{\bf I}+ \frac{1}{T^2} {\bf G}_{diag}| \nonumber  \\
& = {\sf tr} \log \left( {\bf I}+{\bf G}_{diag} \right) - {\sf tr} \log \left( {\bf I}+ \frac{1}{T^2} {\bf G}_{diag} \right) \nonumber \\
& = \sum_{i}{\log \left( 1+g_{ii} \right) } - \sum_{i}{\log \left( 1+\frac{1}{T^2} g_{ii} \right)} \nonumber \\
& = \sum_{i}{\log \left( \frac{1+g_{ii}}{1+\frac{1}{T^2}g_{ii}} \right)} \nonumber \\
& \overset{(a)}{\le} \sum_{i}{\lim_{g_{ii} \rightarrow \infty} \log \left( \frac{1+g_{ii}}{1+\frac{1}{T^2}g_{ii}} \right) } \nonumber \\
& = \sum_{i}{\log T^2} \le 2 |\mathcal{V}| \log T,
\end{align}
where $(a)$ follows from the fact that $\frac{1+g_{ii}}{1+\frac{1}{T^2}g_{ii}}$ is monotonically increasing in $g_{ii}$.

Now suppose that $\min_{\Omega \in \Lambda_D}{\log|{\bf I} + \frac{P}{T^2\sigma^2} {\bf G}_{\Omega} {\bf G}^{\ast}_{\Omega}|} = \log|{\bf I} + \frac{P}{T^2\sigma^2} {\bf G}_{\Omega^\prime} {\bf G}^{\ast}_{\Omega^\prime}|$. Hence, from (\ref{eq:monotone}), we have
\begin{align}
\label{eq:maxmin}
& \min_{\Omega \in \Lambda_D}{\log|{\bf I} + \frac{P}{\sigma^2} {\bf G}_{\Omega} {\bf G}^{\ast}_{\Omega}|} - \min_{\Omega \in \Lambda_D}{\log|{\bf I} + \frac{P}{T^2\sigma^2} {\bf G}_{\Omega} {\bf G}^{\ast}_{\Omega}|} \nonumber \\
& = \min_{\Omega \in \Lambda_D}{\log|{\bf I} + \frac{P}{\sigma^2} {\bf G}_{\Omega} {\bf G}^{\ast}_{\Omega}|} -  \log|{\bf I} + \frac{P}{T^2\sigma^2} {\bf G}_{\Omega^\prime} {\bf G}^{\ast}_{\Omega^\prime}| \nonumber \\
& \le \log|{\bf I} + \frac{P}{\sigma^2} {\bf G}_{\Omega^\prime} {\bf G}^{\ast}_{\Omega^\prime}| - \log|{\bf I} + \frac{P}{T^2\sigma^2} {\bf G}_{\Omega^\prime} {\bf G}^{\ast}_{\Omega^\prime}| \nonumber \\
& \overset{(a)}{\le} 2 |\mathcal{V}| \log T,
\end{align}
where (a) follows from (\ref{eq:monotone}). Hence, from (\ref{eq:appendix1}) and (\ref{eq:appendix2}) we have
\begin{align}
& C(\sigma^2,P) - C(T\sigma^2,P/T) \le \min_{\Omega \in \Lambda_D}{\log|{\bf I} + \frac{P}{\sigma^2} {\bf G}_{\Omega} {\bf G}^{\ast}_{\Omega}|} \nonumber \\
& - \min_{\Omega \in \Lambda_D}{\log|{\bf I} + \frac{P}{T^2\sigma^2} {\bf G}_{\Omega} {\bf G}^{\ast}_{\Omega}|} + 17 |\mathcal{V}|  \overset{(a)}{\le} |\mathcal{V}| \left( 2 \log T + 17 \right),
\end{align}
where $(a)$ follows from (\ref{eq:maxmin}). Therefore, we get
\begin{align}
C(\sigma^2,P) - \tau \le C(T\sigma^2,P/T) \le C(\sigma^2,P),
\end{align}
where $\tau = |\mathcal{V}| \left( 2 \log T + 17 \right)$ is a constant independent of channel gains.
\end{proof}



\bibliographystyle{ieeetr}
\bibliography{bib_partial}

\end{document}